\documentclass{article}
\usepackage[utf8]{inputenc}
\usepackage[dvipsnames]{xcolor}
\usepackage[T1]{fontenc}
\usepackage[english]{babel}
\usepackage{amsmath,amsthm,enumerate}
\usepackage{amssymb}
\usepackage[colorlinks=false, linktocpage=true]{hyperref}
\usepackage{vmargin}
\usepackage{url}
\usepackage{stmaryrd}
\usepackage{tikz}
\usetikzlibrary{decorations.pathreplacing}
\usepackage{authblk}
\usepackage{adjustbox}
\usepackage{subfig}

\usepackage{algorithm}
\usepackage{algorithmic}
\usepackage{comment}

\usepackage{tikz,pgf}
\usepackage{subfig}
\usepackage{adjustbox}
\usetikzlibrary{shapes.geometric}
\usetikzlibrary{decorations.markings}
\usetikzlibrary{math}
\tikzset{->-/.style={decoration={
  markings,
  mark=at position .5 with {\arrow[scale=1]{stealth}}},postaction={decorate}}}
\tikzset{->/.style={decoration={
  markings,
  mark=at position 1 with {\arrow[scale=1]{stealth}}},postaction={decorate}}}

\usepackage{amsthm,amsmath}
\usepackage[obeyFinal]{todonotes}
\usepackage{xspace}
\usepackage{nicefrac}
\usepackage{thmtools}
\usepackage{thm-restate}

\newtheorem{theorem}{Theorem}
\newtheorem{lemma}{Lemma}
\newtheorem{claim}{Claim}
\newtheorem{corollary}{Corollary}

\newtheorem{observation}{Observation}

\newenvironment{proofclaim}{\noindent{\em Proof of the claim.}}{\qedclaim}
\newcommand{\qedclaim}{\hfill $\diamond$ \medskip}

\newcommand{\turn}{s}

\usepackage{bibentry}
\usepackage{cleveref}

\usepackage{complexity}
\newcommand{\MAPF}{\textsc{Multiagent Path Finding}\xspace}
\newcommand{\MAPFShort}{\textsc{MAPF}\xspace}
\newcommand{\MAPFSwap}{\textsc{Multiagent Path Finding Swaps}\xspace}
\newcommand{\MAPFSwapShort}{\textsc{MAPFS}\xspace}
\newcommand{\MAPFNoSwap}{\textsc{Multiagent Path Finding No Swaps}\xspace}
\newcommand{\MAPFNoSwapsShort}{\textsc{MAPFNS}\xspace}

\newcommand{\vc}{\ensuremath{\operatorname{vc}}}
\newcommand{\tw}{\ensuremath{\operatorname{tw}}}
\newcommand{\cw}{\ensuremath{\operatorname{cw}}}
\newcommand{\dst}{\ensuremath{\operatorname{dist}}}

\newenvironment{keywords}{
       \list{}{\advance\topsep by0.35cm\relax\small
       \leftmargin=1cm
       \labelwidth=0.35cm
       \listparindent=0.35cm
       \itemindent\listparindent
       \rightmargin\leftmargin}\item[\hskip\labelsep
                                     \bfseries Keywords:]}
     {\endlist}

\begin{document}

\title{Exact Algorithms and Lowerbounds for Multiagent Pathfinding: Power of Treelike Topology\thanks{This work was co-funded by the European Union under the project Robotics and advanced industrial production (reg.\ no.\ CZ.02.01.01/00/22\_008/0004590).
JMK was additionally supported by the Grant Agency of the Czech Technical University in Prague, grant No.\ SGS23/205/OHK3/3T/18. Furthermore, an extended abstract is accepted in AAAI~$'24$.}}

\author[1]{Foivos Fioravantes}
\author[1]{Du\v{s}an Knop}
\author[1]{Jan Maty\'a\v{s} K\v{r}i\v{s}\'{t}an}
\author[1]{Nikolaos Melissinos}
\author[1]{Michal Opler}
\affil[1]{Department of Theoretical Computer Science, Faculty of Information Technology, Czech Technical University in Prague, Prague, Czech Republic}

\date{}

\maketitle

\begin{abstract}
In the \MAPF problem, we focus on efficiently finding non-colliding paths for a set of $k$ agents on a given graph $G$, where each agent seeks a path from its source vertex to a target.
An important measure of the quality of the solution is the length of the proposed schedule~$\ell$, that is, the length of a longest path (including the waiting time).
In this work, we propose a systematic study under the parameterized complexity framework. The hardness results we provide align with many heuristics used for this problem, whose running time could potentially be improved based on our fixed-parameter tractability results.

We show that \MAPFShort{} is \W[1]-hard with respect to~$k$ (even if~$k$ is combined with the maximum degree of the input graph).
The problem remains \NP-hard in planar graphs even if the maximum degree and the makespan~$\ell$ are fixed constants.
On the positive side, we show an FPT algorithm for $k+\ell$.

As we delve further, the structure of~$G$ comes into play.
We give an \FPT{} algorithm for parameter $k$ plus the diameter of the graph~$G$.
The \MAPFShort{} problem is \W[1]-hard for cliquewidth of $G$ plus~$\ell$ while it is \FPT{} for treewidth of~$G$ plus~$\ell$.
\end{abstract}

{}

\section{Introduction}
In this paper, we study the \MAPF{} (\MAPFShort) problem.
\MAPFShort{} has many real-world applications, e.g., in warehouse management~\cite{WurmanDM08,LiTKDKK21}, airport towing~\cite{MorrisPLMMKK16}, autonomous vehicles, robotics~\cite{VelosoBCR15}, digital entertainment~\cite{MaYCKK17}, and computer games~\cite{SnapeGBLM12}.
In \MAPFShort{} the task is to find non-colliding paths for a set of $k$ agents on a given graph~$G$, where each agent seeks a path from its source vertex to a target.

This important problem has been, to the best of our knowledge, formally introduced about 20 years ago and attracted many researchers since then.
Nowadays, there are numerous variants of the formal model for the problem; see~\cite{SternSFKMWLACKBB19}.
It is not surprising that the \MAPFShort{} problem is in general \NP-complete~\cite{Surynek10}; the hardnes holds even in planar graphs~\cite{Yu16}.
Therefore, there are numerous heuristic approaches that allow us to efficiently obtain a useful solution for the given input; see, e.g., the survey~\cite{Stern19} for overview of such results.
A multitude of techniques was used to tackle \MAPFShort---A*-based algorithms~\cite{HartNR68}, SAT-based algorithms~\cite{SurynekSFB17}, scheduling approach~\cite{BartakSV18}, SMT-solvers~\cite{Surynek19,Surynek20}, to name just a few.
Our work focuses on exact algorithms, that is, the aim is to return an optimal solution and the central question is on which kinds of inputs this can be done in an efficient way and where this is unlikely.
While doing so, we initiate parameterised analysis of the \MAPF{} problem with the focus on natural and structural parameters.
Note that the \MAPFShort{} problem features two natural parameters---the number of agents~$k$ and the total length of the schedule~$\ell$ (also called \emph{makespan}).
It is worth noting that two versions of \MAPFShort{} are usually studied in the literature---one allows swapping two agents along an edge while this is prohibited in the other.
Roughly speaking, the first version treats the input topology as a bidirected graph while the other as undirected.
Most of our algorithmic results are independent of allowing or forbidding swaps; therefore, we reserve the name \MAPF{} for the version where we do not care about swaps.
If the result holds only with swaps allowed, we use \MAPFSwap{} to refer to that specific version fo the problem while we use \MAPFNoSwap{} for the version where we explicitly do not allow swaps.

\subsection*{Our Contribution}
It is well-known that \MAPFShort{} can be reduced to a shortest-path problem in a graph known as the $k$-agent search space.
Then, one can apply any algorithm to find a shortest path in this graph to find a solution to the original problem, e.g., using A*.
The main downside of this approach is the size of the search space.
We first show that it is unlikely to desing an efficient pruning algorithm for the $k$-agent search space:

\begin{restatable}{theorem}{thmwhardnumberagentsDelta}\label{thm:w-hardness-number-agents}
  The \MAPF problem is \W[1]-hard parameterised by the number of agents~$k$ plus~$\Delta(G)$ the maximum degree of the graph~$G$.
\end{restatable}

Indeed, the above observation yields the following.
\begin{theorem}
  The \MAPF{} problem is in \XP{} parameterised by the number of agents~$k$.
\end{theorem}

We continue the study of classical complexity of \MAPFShort{} with focus on the structure of the input graph~$G$.
Driven by some applications, one is interested in specific graph classes such as planar graphs or trees.
Sadly, we show that the \MAPFShort{} problem remains intractable in both these graph classes even if some other parts of the input are fixed constants.

\begin{restatable}{theorem}{thmMAPFisNPCPlanarMakespanDelta}\label{thm:MAPFisNPC:PlanarMakespanDelta}
  The \MAPFSwap{} problem remains \NP-complete even if the input graph~$G$ is planar, $\ell = 3$, and $\Delta(G) = 4$.
\end{restatable}

It was known~\cite{MaTSKK16} that \MAPF{} is \NP-complete when $\ell=3$ (in general graphs).
The \NP-hardness for planar graphs~\cite{Yu16} was recently improved to constant degree~\cite{EibenGK23}; there $\ell=26$ and $\Delta(G) = 4$.
However, \Cref{thm:MAPFisNPC:PlanarMakespanDelta} features all three of these properties.
Furthermore, we believe that our reduction is simpler than the one of Eiben et al.
It is not hard to see that instance with makespan~$2$ and swaps are allowed can be solved in polynomial time via a reduction to Hall's Marriage Theorem (which is used to find suitable middle point of all of the paths from sources to destinations).
Thus, with the above theorem we get a tight dichotomy result in classical complexity.

Moreover we observe that, surprisingly, when swapping is not allowed, the problem becomes hard even for $\ell=2$.
Note that for $\ell=1$ it is trivial.

\begin{restatable}{theorem}{thmMAPFisNPCPlanarMakespanTwo}\label{thm:NP-hard-makespan-2}
 The \MAPFNoSwap{} problem remains \NP-complete even if the input graph~$G$ is planar, $\ell = 2$ and $\Delta(G) = 5$.
\end{restatable}

The proof of Theorem~\ref{thm:NP-hard-makespan-2} is achieved by slightly adjusting the gadgets used in the proof of Theorem~\ref{thm:MAPFisNPC:PlanarMakespanDelta}.

\begin{restatable}{theorem}{thmMAPFisNPCTreesDelta}\label{thm:np-hard-trees}
  The \MAPFNoSwap{} problem remains \NP-complete even if $G$ is a tree of maximum degree $\Delta(G)=5$.
\end{restatable}

Note that the \MAPFSwap was recently shown to be \NP-complete for the case where the imput graph is a tree and the number of agents is equal to the number of vertices of that tree~\cite{ADKLLMRWW22}. Our Theorem~\ref{thm:np-hard-trees} differs from this work as, apart from treating the non-swapping version, it also tackles the problem for an arbitrary number of agents.

At this point, we see that none of the standard parameters---$k$ or $\ell$---is a good parameter alone.
We complement this by showing that the combination of the two parameters results in fixed-parameter tractability.
\begin{theorem}\label{thm:fpt-agents-makespan}
  The \MAPF{} problem is in \FPT{} parameterised by the number of agents~$k$ plus the makespan~$\ell$.
\end{theorem}

In the rest of the paper, we seek a good structural companion parameters to either of these.
We begin with a very restrictive parameter---the vertex cover number that constitutes a rather simple starting point for our more general results.
Here, we exploit the fact that many agents behave in a same way, i.e., they are almost anonymous.
Furthermore, we can prune the input graph~$G$, so that its size is bounded in terms of parameters (i.e., we provide a kernel).
\begin{theorem}\label{thm:MAPF:FPTk+vc}
  The \MAPF{} problem is in \FPT{} parameterised by the number of agents~$k$ plus the vertex cover number~$\vc(G)$.
\end{theorem}

It is known, that graphs of bounded vertex cover number have bounded diameter---the length of a longest shortest path in the graph (in a connected component).
We strengthen the above result by proving that diameter is a strong companion to~$k$.
Note that this yields tractability for many structural parameters including MIM-width as it implies bounded diameter.
The proof uses the fact that if a graph has many vertices (unbounded in terms of $k$ and the diameter) and small dimater, it must contain a vertex of large degree.
We use such a vertex as a hub and prove that if we first route all agents to the neighborhood of that vertex and then to their respective destinations, we obtain routes with makespan~$O(k\cdot d)$, where~$d$ is the diameter.
The theorem then follows from \Cref{thm:fpt-agents-makespan}.
\begin{theorem}
	The \MAPFShort{} problem is in \FPT{} parameterised by the number of agents~$k$ plus the diameter of~$G$.
\end{theorem}

Note that the vertex cover number~$\vc(G)$ also bounds the number of agents that can move simultaneously.
This can be generalized to a number of locally moving agents or, in other words, the size of separators in~$G$.
It is well-known that the size of vertex separators in a graph is related to a graph parameter treewidth. However, in view of our Theorem~\ref{thm:np-hard-trees}, there is little hope to conceive an efficient algorithm parameterized just by the treewidth of the input graph. To remedy this, we additionally parameterize by~$\ell$ to get tractability of \MAPFShort.

\begin{theorem}
  The \MAPF{} problem is in \FPT{} parameterised by the treewidth $\tw(G)$ plus the makespan~$\ell$.
\end{theorem}

Cliquewidth is a parameter further generalizing both vertex cover number and treewidth~\cite{CourcelleO00}.
However, the same additional parameterization as above leads here to intractability.

\begin{restatable}{theorem}{thmMAPFisWHbyCwplusMakespan}\label{thm:w-hardness-cw-makespan}
  The \MAPF problem is in \W[1]-hard parameterised by the cliquewidth $\cw(G)$ plus the makespan~$\ell$.
\end{restatable}

\section{Preliminaries}
For integers $m$ and $n$, we denote~$[m:n]$ the set of all integers between $m$ and $n$, that is, $[m:n] = \{ m , m+1, \ldots, n\}$.
We use $[n]$ as a shorthand for $[1:n]$.

Formally, in the \MAPF problem we are given a graph $G=(V,E)$, a set of agents $A$, a positive integer~$\ell$, and two functions $s_0\colon A \rightarrow V$ and $t\colon A \rightarrow V$ such that for any pair $a , b \in A$ where $a \neq b$, $s_0(a) \neq s_0(b)$ and $t(a) \neq t(b)$.
Initially, each agent $a \in A$, is placed on the vertex $s_0(a)$. At specific times, called turns, the agents are allowed to move to a neighboring vertex, but are not obliged to do so. The agents can make at most one move per turn and each vertex can host at most one agent at a given turn. The position of the agents in the end of turn $i$ (after the agents have moved) is given by a function $s_i\colon A \rightarrow V$.

We consider two versions of the problem; in \MAPFSwap (\MAPFSwapShort) we allow \emph{swaps}, \emph{i.e.}, two agents to move through the same edge during the same turn, while in \MAPFNoSwap (\MAPFNoSwapsShort) we do not allow swaps. In the first case, given  $s_{i-1}(a)$ for every agent $a \in A$, \emph{i.e.}, the positions of the agents at the turn $i-1$, the positions $s_i$ are considered \emph{feasible} if $s_{i}(a)$ is a neighbor of $s_{i-1}(a)$ in $G$, for every agent $a \in A$. In the second case, the positions $s_i$ are feasible if, in addition to the previous, there is no pair of agents $a,b \in A$ such that $s_{i}(a) = s_{i-1}(b)$ and $s_{i}(b) = s_{i-1}(a)$.

We say that a sequence $s_1,\ldots,s_{\ell}$ is a solution of $\langle G, A, s_0, t, \ell \rangle$ if $s_i$ is considered feasible for all $i \in [\ell]$ and $s_\ell = t$. Also, a feasible solution $s_1,\ldots,s_{\ell}$ has \emph{makespan}~$\ell$.
Our goal is to decide if there exists a solution of makespan~$\ell$.

\paragraph*{Parametrized Complexity.}
Parametrized complexity is a computational paradigm that extends classical measures of time (and space) complexity, aiming to examine the computational complexity of problems with respect to a secondary measure---the parameter.
Formally, a parameterized problem is a set of instances $(x,k) \in \Sigma^* \times \mathbb{N}$, where $k$ is called the parameter of the instance.
A parameterized problem is \emph{fixed-parameter tractable} if it can be determined in $f(k)\cdot\operatorname{poly}(|x|)$ time for an arbitrary computable function $f\colon \mathbb{N}\to\mathbb{N}$.
Such a problem then belongs to the class \FPT.
A parameterized problem is \emph{slicewise polynomial} if it can be determined in $|x|^{f(k)}$ time for a computable function $f\colon \mathbb{N}\to\mathbb{N}$.
Such a problem then belongs to the class \XP.
A problem is presumably not in \FPT{} if it is shown to be \W[1]-hard (by a parameterized reduction).
We refer the interested reader to now classical monographs~\cite{CyganFKLMPPS15,Niedermeier06,FlumG06,DowneyF13} for a more comprehensive introduction to this topic.

\paragraph*{Structural Parameters.}
Let $G=(V,E)$ be a graph. A set $U\subseteq V$ is a \emph{vertex cover} if for every edge $e\in E$ it holds that $U\cap e \not= \emptyset$. The \emph{vertex cover number} of $G$, denoted $\vc(G)$, is the minimum size of a vertex cover of $G$.

A \emph{tree-decomposition} of $G$ is a pair $(\mathcal{T},\beta)$, where~$\mathcal{T}$ is a tree rooted at a node $r\in V(\mathcal{T})$, $\beta\colon V(\mathcal{T})\to 2^{V}$ is a function assigning each node $x$ of $\mathcal{T}$ its \emph{bag}, and the following conditions hold:
\begin{itemize}
	\item for every edge $\{u,v\}\in E(G)$ there is a node $x\in V(\mathcal{T})$ such that $u,v\in\beta(x)$ and
	\item for every vertex $v\in V$, the set of nodes $x$ with $v\in\beta(x)$ induces a connected subtree of $\mathcal{T}$.
\end{itemize}
The \emph{width} of a tree-decomposition $(\mathcal{T},\beta)$ is $\max_{x\in V(\mathcal{T})} |\beta(x)|-1$, and the treewidth $\tw(G)$ of a graph $G$ is the minimum width of a tree-decomposition of $G$.
It is known that computing a tree-decomposition of minimum width is fixed-parameter tractable when parameterized by the treewidth~\cite{Kloks94,Bodlaender96}, and even more efficient algorithms exist for obtaining near-optimal tree-decompositions~\cite{KorhonenL23}.

Cliquewidth is a prominent parameter that generalizes treewidth~\cite{CourcelleO00}.
Cliquewidth corresponds to an algebraic expression describing the graph (i.e., how the graph is created from labeled vertices).
We can use (1) introducing a single vertex of any label, (2) disjoint union of two lebeled graphs, (3) introducing edges between \emph{all} pair of vertices of two labels, and (4) changing the label of \emph{all} vertices of a given label (i.e., collapsing a pair of labels).
An expression describes a graph~$G$ if $G$ the final graph given by the expression (after we remove all the labels).
The \emph{width} of an expression is the number of labels it uses.
The cliquewidth of a graph is the minimum width of an expression describing it.

\section{FPT Algorithms}
For all the results presented in this section, we give a precise running time which was not given in the introduction.

\subsection{Bounded number of agents plus makespan}

Let $G = (V,E)$ be an undirected graph and $\ell$ a positive integer.
We capture the movement of agents through $G$ via a directed graph with vertices representing positions in both place and time.
A \emph{time-expanded graph of $G$ with respect to $\ell$}, denoted $G_T(\ell)$, is a directed graph with one copy of each vertex $v \in V$ for each time step $i \in [0:\ell]$, i.e., its vertex set is $\{v_j \mid v \in V, j\in [0:\ell]\}$.
The set of vertices $\{v_j \mid v \in V\}$ for any fixed~$j$, is called a \emph{layer}.
For every edge $\{u,v\} \in V$ and every $i \in [\ell]$, the graph $G_T(\ell)$ contains two arcs $(u_{i-1}, v_i)$ and $(v_{i-1}, u_i)$.
Moreover for each vertex $v \in V$ and $i \in [\ell]$, the graph $G_T(\ell)$ contains an arc $(v_{i-1}, v_i)$.
Vertex-disjoint paths in $G_T(\ell)$ capture exactly the valid movements of agents when swaps are allowed.

We also introduce a modified version of this graph for the swap-free version of \MAPF.
A \emph{swap-free time-expanded graph of $G$ with respect to $\ell$}, denoted $G^\star_T(\ell)$, is a directed graph consisting of $2\ell+1$ layers, that we again distinguish using subscripts.
It contains a copy $v_i$ of vertex $v\in V$ for every $i \in [0:2\ell]$.
Moreover $G^\star_T(\ell)$ contains a vertex $e_{2i-1}$ for every edge $e \in E$ and every $i \in [\ell]$.
Similarly to $G_T(\ell)$, $G^\star_T(\ell)$ contains an arc $(v_{i-1}, v_i)$ for every vertex $v \in V$ and every $i \in [2\ell]$.
And finally for every edge $e=(u,v)$ in $E$ and every $i\in [\ell]$, there are four arcs $(u_{2i-2},e_{2i-1})$, $(v_{2i-2},e_{2i-1})$, $(e_{2i-1},u_{2i})$ and $(e_{2i-1},v_{2i})$ in $G^\star_T(\ell)$.
Again, it is straightforward to see that vertex-disjoint paths in $G^\star_T(\ell)$ capture exactly the valid movements of agents when swaps are disallowed.

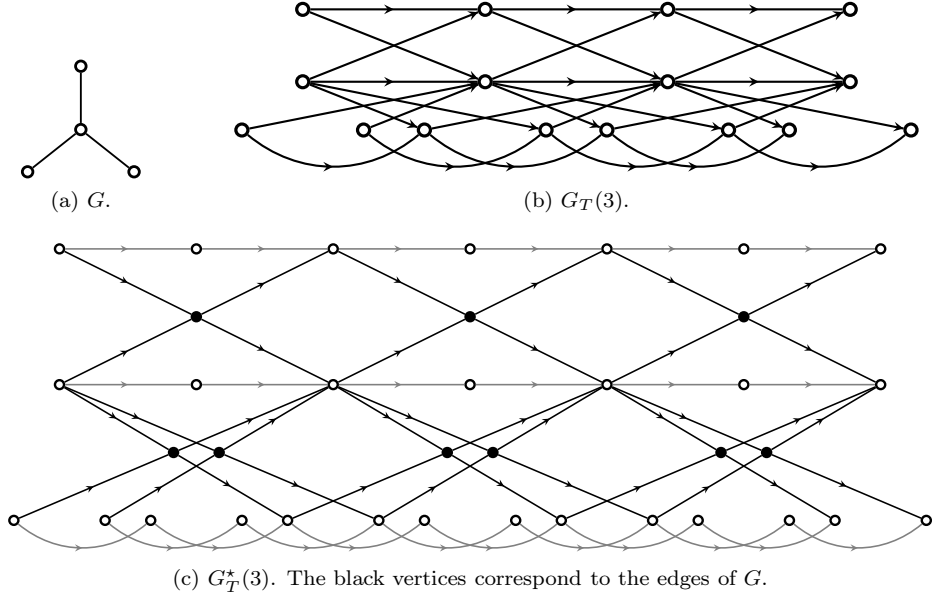
\begin{figure}[!t]
\centering

\subfloat[$G$.]{
\scalebox{0.7}{
\begin{tikzpicture}[inner sep=0.7mm]

\node[draw,black,circle,line width=1.5pt,fill=white] (v1) at (1.5,2){};
\node[draw,black,circle,line width=1.5pt,fill=white] (v2) at (1.5,0.8){};
\node[draw,black,circle,line width=1.5pt,fill=white] (v3) at (2.5,0){};
\node[draw,black,circle,line width=1.5pt,fill=white] (v4) at (0.5,0){};

\draw[line width=1pt,black] (v1) -- (v2);
\draw[line width=1pt,black] (v2) -- (v3);
\draw[line width=1pt,black] (v2) -- (v4);

\end{tikzpicture}
}
}\hspace{20pt}
\subfloat[$G_T(3)$.]{
\scalebox{0.8}{
\begin{tikzpicture}[inner sep=0.7mm]

\node[draw,black,circle,line width=1.5pt,fill=white] (v1) at (1.5,2){};
\node[draw,black,circle,line width=1.5pt,fill=white] (v2) at (1.5,0.8){};
\node[draw,black,circle,line width=1.5pt,fill=white] (v3) at (2.5,0){};
\node[draw,black,circle,line width=1.5pt,fill=white] (v4) at (0.5,0){};

\node[draw,black,circle,line width=1.5pt,fill=white] (v5) at (4.5,2){};
\node[draw,black,circle,line width=1.5pt,fill=white] (v6) at (4.5,0.8){};
\node[draw,black,circle,line width=1.5pt,fill=white] (v7) at (5.5,0){};
\node[draw,black,circle,line width=1.5pt,fill=white] (v8) at (3.5,0){};

\node[draw,black,circle,line width=1.5pt,fill=white] (v9) at (7.5,2){};
\node[draw,black,circle,line width=1.5pt,fill=white] (v10) at (7.5,0.8){};
\node[draw,black,circle,line width=1.5pt,fill=white] (v11) at (8.5,0){};
\node[draw,black,circle,line width=1.5pt,fill=white] (v12) at (6.5,0){};

\node[draw,black,circle,line width=1.5pt,fill=white] (v13) at (10.5,2){};
\node[draw,black,circle,line width=1.5pt,fill=white] (v14) at (10.5,0.8){};
\node[draw,black,circle,line width=1.5pt,fill=white] (v15) at (11.5,0){};
\node[draw,black,circle,line width=1.5pt,fill=white] (v16) at (9.5,0){};

\draw[->-,line width=1pt,black] (v1) -- (v5);
\draw[->-,line width=1pt,black] (v5) -- (v9);
\draw[->-,line width=1pt,black] (v9) -- (v13);

\draw[->-,line width=1pt,black] (v2) -- (v6);
\draw[->-,line width=1pt,black] (v6) -- (v10);
\draw[->-,line width=1pt,black] (v10) -- (v14);

\draw[->-,line width=1,black] (v4) to[out=-40,in=-140] (v8);
\draw[->-,line width=1,black] (v3) to[out=-40,in=-140] (v7);

\draw[->-,line width=1,black] (v8) to[out=-40,in=-140] (v12);
\draw[->-,line width=1,black] (v7) to[out=-40,in=-140] (v11);

\draw[->-,line width=1,black] (v11) to[out=-40,in=-140] (v15);
\draw[->-,line width=1,black] (v12) to[out=-40,in=-140] (v16);

\draw[->,line width=1pt,black] (v1) -- (v6);
\draw[->,line width=1pt,black] (v2) -- (v5);
\draw[->,line width=1pt,black] (v2) -- (v8);
\draw[->,line width=1pt,black] (v2) -- (v7);
\draw[->,line width=1pt,black] (v3) -- (v6);
\draw[->,line width=1pt,black] (v4) -- (v6);

\draw[->,line width=1pt,black] (v5) -- (v10);
\draw[->,line width=1pt,black] (v6) -- (v9);
\draw[->,line width=1pt,black] (v6) -- (v12);
\draw[->,line width=1pt,black] (v6) -- (v11);
\draw[->,line width=1pt,black] (v7) -- (v10);
\draw[->,line width=1pt,black] (v8) -- (v10);

\draw[->,line width=1pt,black] (v9) -- (v14);
\draw[->,line width=1pt,black] (v10) -- (v13);
\draw[->,line width=1pt,black] (v10) -- (v16);
\draw[->,line width=1pt,black] (v10) -- (v15);
\draw[->,line width=1pt,black] (v11) -- (v14);
\draw[->,line width=1pt,black] (v12) -- (v14);

\end{tikzpicture}
}
}

\subfloat[$G^\star_T(3)$. The black vertices correspond to the edges of $G$.]{
\scalebox{0.6}{
\begin{tikzpicture}[inner sep=0.7mm]
\node[draw,black,circle,line width=1.5pt,fill=white] (v1) at (1.5,6){};
\node[draw,black,circle,line width=1.5pt,fill=white] (v2) at (1.5,3){};
\node[draw,black,circle,line width=1.5pt,fill=white] (v3) at (2.5,0){};
\node[draw,black,circle,line width=1.5pt,fill=white] (v4) at (0.5,0){};

\node[draw,black,circle,line width=1.5pt,fill=white] (v5) at (4.5,6){};
\node[draw,black,circle,line width=1.5pt,fill=white] (v6) at (4.5,3){};
\node[draw,black,circle,line width=1.5pt,fill=white] (v7) at (5.5,0){};
\node[draw,black,circle,line width=1.5pt,fill=white] (v8) at (3.5,0){};

\node[draw,black,circle,line width=1.5pt,fill=black] (e11) at (4.5,4.5){};
\node[draw,black,circle,line width=1.5pt,fill=black] (e12) at (4,1.5){};
\node[draw,black,circle,line width=1.5pt,fill=black] (e13) at (5,1.5){};

\node[draw,black,circle,line width=1.5pt,fill=white] (v9) at (7.5,6){};
\node[draw,black,circle,line width=1.5pt,fill=white] (v10) at (7.5,3){};
\node[draw,black,circle,line width=1.5pt,fill=white] (v11) at (8.5,0){};
\node[draw,black,circle,line width=1.5pt,fill=white] (v12) at (6.5,0){};

\node[draw,black,circle,line width=1.5pt,fill=white] (v13) at (10.5,6){};
\node[draw,black,circle,line width=1.5pt,fill=white] (v14) at (10.5,3){};
\node[draw,black,circle,line width=1.5pt,fill=white] (v15) at (11.5,0){};
\node[draw,black,circle,line width=1.5pt,fill=white] (v16) at (9.5,0){};

\node[draw,black,circle,line width=1.5pt,fill=black] (e21) at (10.5,4.5){};
\node[draw,black,circle,line width=1.5pt,fill=black] (e22) at (10,1.5){};
\node[draw,black,circle,line width=1.5pt,fill=black] (e23) at (11,1.5){};

\node[draw,black,circle,line width=1.5pt,fill=white] (v17) at (13.5,6){};
\node[draw,black,circle,line width=1.5pt,fill=white] (v18) at (13.5,3){};
\node[draw,black,circle,line width=1.5pt,fill=white] (v19) at (14.5,0){};
\node[draw,black,circle,line width=1.5pt,fill=white] (v20) at (12.5,0){};

\node[draw,black,circle,line width=1.5pt,fill=white] (v21) at (16.5,6){};
\node[draw,black,circle,line width=1.5pt,fill=white] (v22) at (16.5,3){};
\node[draw,black,circle,line width=1.5pt,fill=white] (v23) at (17.5,0){};
\node[draw,black,circle,line width=1.5pt,fill=white] (v24) at (15.5,0){};

\node[draw,black,circle,line width=1.5pt,fill=black] (e31) at (16.5,4.5){};
\node[draw,black,circle,line width=1.5pt,fill=black] (e32) at (16,1.5){};
\node[draw,black,circle,line width=1.5pt,fill=black] (e33) at (17,1.5){};

\node[draw,black,circle,line width=1.5pt,fill=white] (v25) at (19.5,6){};
\node[draw,black,circle,line width=1.5pt,fill=white] (v26) at (19.5,3){};
\node[draw,black,circle,line width=1.5pt,fill=white] (v27) at (20.5,0){};
\node[draw,black,circle,line width=1.5pt,fill=white] (v28) at (18.5,0){};

\draw[->-,line width=1pt,gray] (v1) -- (v5);
\draw[->-,line width=1pt,gray] (v5) -- (v9);
\draw[->-,line width=1pt,gray] (v9) -- (v13);
\draw[->-,line width=1pt,gray] (v13) -- (v17);
\draw[->-,line width=1pt,gray] (v17) -- (v21);
\draw[->-,line width=1pt,gray] (v21) -- (v25);

\draw[->-,line width=1pt,gray] (v2) -- (v6);
\draw[->-,line width=1pt,gray] (v6) -- (v10);
\draw[->-,line width=1pt,gray] (v10) -- (v14);
\draw[->-,line width=1pt,gray] (v14) -- (v18);
\draw[->-,line width=1pt,gray] (v18) -- (v22);
\draw[->-,line width=1pt,gray] (v22) -- (v26);

\draw[->-,line width=1,gray] (v4) to[out=-40,in=-140] (v8);
\draw[->-,line width=1,gray] (v3) to[out=-40,in=-140] (v7);

\draw[->-,line width=1,gray] (v8) to[out=-40,in=-140] (v12);
\draw[->-,line width=1,gray] (v7) to[out=-40,in=-140] (v11);

\draw[->-,line width=1,gray] (v11) to[out=-40,in=-140] (v15);
\draw[->-,line width=1,gray] (v12) to[out=-40,in=-140] (v16);

\draw[->-,line width=1,gray] (v15) to[out=-40,in=-140] (v19);
\draw[->-,line width=1,gray] (v16) to[out=-40,in=-140] (v20);

\draw[->-,line width=1,gray] (v19) to[out=-40,in=-140] (v23);
\draw[->-,line width=1,gray] (v20) to[out=-40,in=-140] (v24);

\draw[->-,line width=1,gray] (v23) to[out=-40,in=-140] (v27);
\draw[->-,line width=1,gray] (v24) to[out=-40,in=-140] (v28);

\draw[->-,line width=1,black] (v1) -- (e11);
\draw[->-,line width=1,black] (v2) -- (e11);
\draw[->-,line width=1,black] (e11) -- (v9);
\draw[->-,line width=1,black] (e11) -- (v10);

\draw[->-,line width=1,black] (v2) -- (e12);
\draw[->-,line width=1,black] (v4) -- (e12);
\draw[->-,line width=1,black] (e12) -- (v12);
\draw[->-,line width=1,black] (e12) -- (v10);

\draw[->-,line width=1,black] (v2) -- (e13);
\draw[->-,line width=1,black] (v3) -- (e13);
\draw[->-,line width=1,black] (e13) -- (v11);
\draw[->-,line width=1,black] (e13) -- (v10);

\draw[->-,line width=1,black] (v9) -- (e21);
\draw[->-,line width=1,black] (v10) -- (e21);
\draw[->-,line width=1,black] (e21) -- (v17);
\draw[->-,line width=1,black] (e21) -- (v18);

\draw[->-,line width=1,black] (v10) -- (e22);
\draw[->-,line width=1,black] (v12) -- (e22);
\draw[->-,line width=1,black] (e22) -- (v20);
\draw[->-,line width=1,black] (e22) -- (v18);

\draw[->-,line width=1,black] (v10) -- (e23);
\draw[->-,line width=1,black] (v11) -- (e23);
\draw[->-,line width=1,black] (e23) -- (v19);
\draw[->-,line width=1,black] (e23) -- (v18);

\draw[->-,line width=1,black] (v17) -- (e31);
\draw[->-,line width=1,black] (v18) -- (e31);
\draw[->-,line width=1,black] (e31) -- (v25);
\draw[->-,line width=1,black] (e31) -- (v26);

\draw[->-,line width=1,black] (v18) -- (e32);
\draw[->-,line width=1,black] (v20) -- (e32);
\draw[->-,line width=1,black] (e32) -- (v28);
\draw[->-,line width=1,black] (e32) -- (v26);

\draw[->-,line width=1,black] (v18) -- (e33);
\draw[->-,line width=1,black] (v19) -- (e33);
\draw[->-,line width=1,black] (e33) -- (v27);
\draw[->-,line width=1,black] (e33) -- (v26);
\end{tikzpicture}
}
}

\caption{An example of a graph $G$, together with $G_T(3)$ and $G^\star_T(3)$.}
\label{figure:example-time-expanded}
\end{figure}

\begin{observation}
\label{obs:time-exp-graph}
Let $\mathcal{I} = \langle G, A, s_0, t ,\ell\rangle$ be an instance of \MAPFShort. Then $\mathcal{I}$ is a yes-instance of \MAPFSwapShort (\MAPFNoSwapsShort resp.) if and only if there exists a set of directed pairwise vertex-disjoint paths in $G_T(\ell)$ ($G^\star_T(\ell)$ resp.) connecting all pairs $\{(s_0(a)_0, t(a)_\ell) \mid a \in A\}$ ($\{(s_0(a)_0, t(a)_{2\ell}) \mid a \in A\}$ resop.), where the second index is used to denote the corresponding layer of $G_T(\ell)$ ($G^*_T(\ell)$ resp.).

\end{observation}

\begin{theorem}
\label{thm:}
The \MAPF{} problem can be solved by an FPT-algorithm parameterised by the number of agents $k$ plus the makespan $\ell$ in $O(2^{O(k\cdot \ell)} \cdot m \cdot \log n)$  time.
\end{theorem}
\begin{proof}
We reduce to the \textsc{Bounded Vertex Directed Multi-terminal disjoint Paths (BVDMP)} problem.
Its input consists of a directed graph $G$, two positive integers $k, d$, and $k$ pairs of vertices $\{(s_i,t_i) \mid i\in [k]\}$ in~$G$; and the task is to decide whether there is a set of $k$ directed vertex-disjoint paths of length at most $d$ connecting all pairs $\{(s_i, t_i) \mid i \in [k]\}$.
The \textsc{BVDMP} problem can be solved by an FPT-algorithm parameterised by $k$ plus $d$ in $O(2^{O(k\cdot d)} \cdot m \cdot \log n)$ time~\cite{GolovachT11}.

Let $\langle G, A, s_0, t, \ell \rangle$ be an instance of \MAPFShort.
Due to Observation~\ref{obs:time-exp-graph}, it is sufficient to test whether there exists a set of $k$ directed pairwise vertex-disjoint paths either in $G_T(\ell)$ or in $G^\star_T(\ell)$ depending on whether we allow swaps or not.
Notice that such paths have length exactly $\ell$ in $G_T(\ell)$ and $2\ell$ in $G_T^\star(\ell)$.
Thus for \MAPFSwapShort, it suffices to invoke the FPT algorithm for \textsc{BVDMP} on the graph $G_T(\ell)$ with the pairs of vertices $\{(s_0(a)_0, t(a)_\ell) \mid a \in A\}$ and setting $d = \ell$.
While for \MAPFNoSwapsShort, we run the same algorithm on the graph $G^\star_T(\ell)$ with the pairs of vertices $\{(s_0(a)_0, t(a)_{2 \ell}) \mid a \in A\}$ and $d = 2\ell$.
The resulting algorithms run in $O(2^{O(k\cdot \ell)} \cdot m \cdot \log n)$ time since both $G_T(\ell)$ and $G_T^\star(\ell)$ have $O(\ell \cdot m)$ edges.
\end{proof}

\subsection{Bounded treewidth plus makespan}

\begin{lemma}
\label{lem:time-exp-tw}
For an undirected graph $G$ and a non-negative integer $p$, the treewidth of both graphs $G_T(\ell)$ and $G^\star_T(\ell)$ is at most $O(\ell \cdot \tw(G))$.
\end{lemma}
\begin{proof}
Let $(T,\beta)$ be a tree decomposition of $G = (V,E)$ of optimal width.
First, let us consider the graph~$G_T(\ell)$.
Let $(T,\beta')$ be the tree decomposition obtained by replacing every occurrence of a vertex $v$ in any bag with its $\ell+1$ copies $v_0, \dots, v_\ell$.
Clearly, $|\beta'(x)| = (\ell+1) \cdot |\beta(x)|$ for every node~$x$ and it is easy to check that all the properties of tree decompositions hold.

Now, let us consider $G_T^\star(\ell)$.
We first modify the tree decomposition $(T, \beta)$ into $(T', \beta')$ such that each edge $e$ has a unique associated node $x'_e$ in $T'$ such that both endpoints of~$e$ lie in $\beta'(x'_e)$.
That is easily achieved by first choosing an arbitrary such node $x_e$ for each edge and then creating its new copy $x'_e$ attached to $x_e$ as a leaf.
We construct a tree decomposition $(T'',\beta'')$ of $G_T^\star(\ell)$ by the following modification of $(T', \beta')$.
As before, we replace every occurrence of a vertex~$v$ in any bag with its $2\ell+1$ copies $v_0, \dots, v_{2\ell}$.
Moreover, for each edge $e \in E$, we add the vertices $e_1, e_3, \dots, e_{2\ell-1}$ to the bag $\beta''(x'_e)$.
Observe that $|\beta''(x)| \le (2\ell+1) \cdot |\beta(x)| + \ell$ for every node $x$ in $T''$.
It is again straightforward to verify the required properties of tree decompositions.
\end{proof}

\begin{theorem}
The \MAPF problem can be solved by an FPT-algorithm parameterised by the treewidth~$w$ of~$G$ plus the makespan~$\ell$ in $(\ell \cdot w)^{O(\ell \cdot w)} \cdot n$ time.
\end{theorem}
\begin{proof}
Let $\langle G, A, s_0, t, \ell \rangle$ be an instance of \MAPFShort.
First, let us consider the variant \MAPFSwapShort where swaps are allowed.
Due to Observation~\ref{obs:time-exp-graph}, it suffices to check whether there exists a set of directed pairwise vertex-disjoint paths in $G_T(\ell)$ connecting all pairs $\{(s_0(a)_0, t(a)_\ell) \mid a \in A\}$.
Moreover by Lemma~\ref{lem:time-exp-tw}, the treewidth of $G_T(\ell)$ is at most $O(\ell\cdot \tw(G))$.
Finding vertex-disjoint paths in a directed graph $G_T(\ell)$ is in FPT parameterised by the treewidth $G_T(\ell)$ via a simple modification of the undirected case~\cite{Scheffler94} accounting for the orientation of the arcs.
The algorithm runs in $(w')^{O(w')} \cdot n$ time, where $w'$ is the treewidth of $G_T(\ell)$.
Thus, we can solve the \MAPF{} problem in $(\ell \cdot w)^{O(\ell \cdot w)} \cdot n$ time, where $w$ is the treewidth of~$G$.

For the swap-free variant \MAPFNoSwapsShort, the algorithm follows via the same argument simply by using the swap-free time-expanded graph $G_T^\star(\ell)$ instead of $G_T(\ell)$.
Observe that $G_T^\star(\ell)$ has $O(\ell \cdot (n+ m))$ vertices and $O(\ell \cdot (n+m))$ edges where $n$, $m$ is the number of vertices and edges in $G$, respectively.
However, we can bound $m$ by $O(w \cdot n)$ since graphs of bounded tree-width are sparse~\cite{Kloks94} and thus the asymptotic bound on the runtime remains unchanged.
\end{proof}

\subsection{Bounded number of agents plus vertex cover number}
\begin{theorem}
The \MAPF problem, admits a kernel of size $O(2^{\vc}k)$, where $\vc$ is the size of a minimum vertex cover of the given graph and $k$ the number of agents.
\end{theorem}

\begin{proof}
    Let $\langle G,A, s_0 , t, \ell\rangle$ be an instance of \MAPFShort problem and $U\subseteq V(G)$ be a minimum vertex cover of $G$. For each subset $S\subset U$, let $V_S$ denote the subset of $V(G) \setminus U $ where $v \in V_S$ if and only if $N(v) = S$. Notice that any pair of vertices $u,v$ that belong in the same set $V_S$ , for some $S\subset U$, are twins. For each $S$, we select a \emph{representative} set of vertices $U_S$, defined as follows: if $|V_S|\le 3k $ then $U_{S} = V_S$, otherwise we select a $U_S$ that satisfies the following properties:
    \begin{itemize}
        \item $U_{S} \subseteq  V_S$,
        \item $U_{S} \supseteq  V_S \cap \{s_0(a), t(a)\mid a \in A \}$ and
        \item $|U_{S} \setminus \{s_0(a), t(a)\mid a \in A \} | = k$.
    \end{itemize}
    There are enough vertices to add so that the last property is satisfied since $|\{s_0(a), t(a)\mid a \in A \}| \le 2k$ and $|V_S| > 3k$. Also, any set that satisfies these conditions has size at most $3k$.
    Now, we create the new instance of \MAPFShort as follows. First, we create $H$, the induced subgraph of $G$ where $H = G[  U \cup \bigcup_{S \subseteq U} U_S]$. The new instance has the same set of agents, with the same starting and terminal positions as the original one. This gives us a new instance $\langle H,A, s_0 , t, \ell\rangle$, where $|V(H)| = |U \cup \bigcup_{S \subseteq U} U_S | = \vc + 2^{\vc} 3k$. We will show that the two instances are equivalent.

    Since $H$ is an induced subgarph of $G$, any feasible solution in $H$ is also a feasible solution in $G$. Therefore, if $\langle H,A, s_0 , t, \ell\rangle$ is a yes-instance, this is also true for $\langle G,A, s_0 , t, \ell \rangle$.
    Next, we show that any feasible solution of $\langle G,A, s_0 , t, \ell \rangle$ can be transformed into a feasible solution of $\langle H, A, s_0 , t, \ell\rangle$.
    Let $s_1, \ldots, s_\ell$ be a feasible sequence for $\langle G,A, s_0 , t, \ell \rangle$. We create a sequence $s_1',\ldots, s_\ell'$ for $\langle H, A, s_0, t, \ell \rangle$ inductively.

    In particular, we will create a sequence $s_1',\ldots, s_\ell'$ where, for any $i \in [\ell]$ and agent $a\in A$, the following holds:
    \begin{itemize}
        \item if $s_i(a) \in U \cup \{s_0(a), t(a)\mid a \in A \}$ then $s_i' (a) = s_i(a)$,
        \item otherwise, $s_{i}' (a) \in U_S$ where $S = N_{G}(s_i(a))$.
    \end{itemize}

    We start the construction with $s_1'$. First we partition the agent set $A$ in to two sets $A_1$ and $A_2$ where $a \in A_1$ if $s_1(a) \in V(H)$ and $a \in A_2$ otherwise. For any agent $a \in A_1$ we set $s_1' (a) = s_1(a)$.
    Then, we select an agent $b \in A_2$; let $s_1(b) = v $ and $S = N_G(b)$.  We are selecting a vertex $u \in U_{S}$ such that $u \neq s_1(a)$ for any $a \in A_1$. We set $s_1'(b) = u$ and we move $b$ form $A_2$ to $A_1$. Notice that, since $v \notin V(H)$, we have enough vertices in $U_S \setminus \{s_0(a), t(a)\mid a \in A \}$ to achieve this.
    We repeat until $A_2$ is empty.

    All the moves described by $s_1'$ are valid in $H$. Indeed,
    notice that for any agent $a \in A$ such that $s_1' (a)= s_1(a)$, this move is obviously valid as it was valid in $G $. Also, for any other agent the move is also valid because the selected $s_1' (a)$ is a twin of $ s_1(a)$.
    Finally, notice that for any agent $a \in A$ such that $s_1' (a) \neq  s_1(a)$, neither of $s_1' (a) $ and $ s_1(a)$ belong to $ U \cup \{s_0(b), t(b)\mid b \in A \}$.

    Assume now that we have created the sub-sequence $s_1', \ldots , s_i'$ for some $i <\ell$, all moves between $s_{j-1}'$ and $s_j'$ are valid for any $j \in [i]$ and that for any $j \in [i]$ and agent $a\in A$,
    \begin{itemize}
        \item if $s_j(a) \in U \cup \{s_0(a), t(a)\mid a \in A \}$ then $s_j' (a) = s_j(a)$,
        \item otherwise, $s_{j}' (a) \in U_S$ where $S = N_{G}(s_j(a))$.
    \end{itemize}

    We create $s_{i+1}'$ using $s_{i+1}$, $s_{i}$ and $s_{i}'$.
    We start by partitioning $A$ in to three sets $A_1$, $A_2$ and $A_3$ such that:
    \begin{itemize}
        \item $A_1 = \{a \in A \mid s_{i+1}(a) \in U \cup \{s_0(b), t(b)\mid b \in A \}\}$,
        \item $A_2 = \{a \in A \setminus A_1 \mid s_{i+1}(a) = s_{i} (a) \}$,
        \item $A_3 =  A \setminus (A_1 \cup A_2) $.
    \end{itemize}
    For any agent $a \in A_1$, we set $s_{i+1}'(a) = s_{i+1}(a)$.
    For any agent $a \in A_2$, we set $s_{i+1}'(a) = s_{i}'(a)$.
    We create an empty set $B$. Then we select an agent $b \in A_3 \setminus B$; let $s_{i+1}(b) = v $ and $S = N_G(v)$.
    We select a vertex $u \in U_{S} \setminus \{s_0(a), t(a)\mid a \in A \}$ such that $u \neq s_{i+1}'(a)$ for any $a \in A_1 \cup A_2 \cup B$.
    We set $s_{i+1}'(b) = u$ and we add $b$ to $B$.
    We repeat until $A_3 \setminus B$ is empty.

    \begin{claim}
        We always have enough vertices in the sets $U_{S}\setminus \{s_0(a), t(a)\mid a \in A \}$ in order to select vertices that have the wanted properties for the agents in the set $A_3$.
    \end{claim}

    Before we prove this claim we need to observe that, for any turn $i \in \ell$ and any agent $a \in A$, $s'_{i}(a)\in U_S$ if and only if $s_{i}(a)\in V_S$. This is obvious if $a\in A_1 \cup A_3$. To prove this for the agents in $A_2$, consider the first turn $j <i$ such that $s_{k}(a) = s_{i}(a)$ for all $j\le k\le i $. During turn $j$, $a$ belonged to $A_1 \cup A_3$ (defined during that turn) therefore it is true.

    \begin{proofclaim}
        Assume that we have an agent $a \in A_3$, $s_{i+1}(a) = v $ and $S = N_G(v)$. We consider two cases, either $|V_S| > 3k$ or not.  In the former, we know that $|U_{S}\setminus \{s_0(c), t(c)\mid c \in A \}| = k$. Therefore, since $|A\setminus\{a\}| = k-1$, even if for all agents $b \in A$ we have that $s_{i+1}'(b) \in U_{S}\setminus \{s_0(c), t(c)\mid c \in A \} $, there exists at least one vertex in $U_{S}\setminus \{s_0(c), t(c)\mid c \in A \}$ that can be assigned to $s_{i+1}'(a)$.

        In the latter cases, let $A' \subseteq A$ be the set of agents where $b \in A'$ if and only if $s_{i+1}(b) \in V_S$ and $s_{i+1}(b) \notin \{s_0(b), t(b)\mid b \in A \}$. Since $s_{i+1}'(b) \in U_S\setminus \{s_0(b), t(b)\mid b \in A \}$ only if $s_{i+1}(b) \in V_S$, we need to prove that $|A'| \le |U_S| \setminus \{s_0(b), t(b)\mid b \in A \}$. This is indeed the case, as we have assumed that $|V_S| \le 3k$, and therefore $U_S = V_S$.
    \end{proofclaim}

    Notice that, for any agent $a \in A$, either $s_{i}'(a) = s_{i}(a)$ or $s_{i}'(a)$ is a twin of $s_{i}(a)$. The same holds for $s_{i+1}'(a)$ and $s_{i+1}(a)$.
    We will show that $s_{i+1}'(a)$ is either equal to $s_{i}'(a)$ or its neighbor.
    Consider the case where $s_{i}(a) = s_{i+1}(a)$. Regardless of whether $a \in A_1$ or $a \in A_2$, we have that $s_{i}'(a) = s_{i+1}'(a)$ which is a valid move.
    Now, consider the case where $s_{i}(a) \neq s_{i+1}(a)$. In this case, we know that $s_{i}(a)$ and $ s_{i+1}(a)$ are neighbors in $G$. Since $s_{i}'(a) = s_{i}(a)$ or $s_{i}'(a)$ is a twin of $s_{i}(a)$ and the same holds for $s_{i+1}'(a)$ and $s_{i+1}(a)$, we can conclude that $s_{i}'(a)$ and $ s_{i+1}'(a)$ are neighbors in $G$. Therefore $s_{i}'(a)$ and $ s_{i+1}'(a)$ are neighbors in $H$ and the move is valid.

    Next, we show that there are no two agents that occupy the same vertex during the same turn. Notice that the the values $s_{i+1}(a)$ for $a\in A_3$ are selected in such a way that no agent of $A_3$ will ever occupy the same vertex as any other agent.
    Therefore we need to show that the same holds true for agents in $A_1 \cup A_2$.
    Assume that this is not true and there exist two agents $a,b \in A_1 \cup A_2$ such that $s_{i+1}'(a)=s_{i+1}'(b)$. We consider three cases.

    \textbf{Case 1. $\boldsymbol{{a,b} \subseteq A_1}$:} Since $s_1,\ldots,s_\ell$ is a feasible solution, we have that $s_{i+1}(a)\neq s_{i+1}(b)$.
    Furthermore, since ${a,b} \subseteq A_1$, we have that $s_{i+1}'(a) =s_{i+1}(a) \neq s_{i+1}(b) = s_{i+1}'(b) $. This is a contradiction.

    \textbf{Case 2. $\boldsymbol{a \in A_1$ and $b \in A_2}$:} Since $s_1,\ldots,s_\ell$ is feasible, we have that $s_{i}(a)\neq s_{i}(b)$ and $s_{i+1}(a)\neq s_{i+1}(b)$.

    We also have that $s_1', \ldots , s_i'$ is feasible and, thus, $s_{i}'(a) \neq s_{i}'(b)$.
    Since $a \in A_1$, we have that $s_{i+1}'(a) = s_{i+1}(a) \in U \cup \{s_0(c), t(c)\mid c \in A \}$.
    Furthermore, we have assumed that $s_{i+1}'(a)=s_{i+1}'(b)$ and $b \in A_2$. It follows that $s_{i}'(b) = s_{i+1}'(b) \in U \cup \{s_0(c), t(c)\mid c \in A \}$.
    By the hypothesis for $s_i'$ and $s_{i}'(b) \in U \cup \{s_0(c), t(c)\mid c \in A \}$ we have that $s_{i}'(b) = s_{i}(b)$. Therefore, $s_{i+1}(b) = s_{i}(b) = s_{i}'(b) = s_{i+1}'(b)= s_{i+1}'(a) = s_{i+1}(a)$. This is a contradiction as $s_{i+1}(a) \neq  s_{i+1}(b)$.

    \textbf{Case 3. $\boldsymbol{{a,b} \subseteq A_2}$:} Since $a \in A_2$, we have that $s_{i}'(a) = s_{i+1}'(a)$. Also, since $a \in A_2$, we have that $s_{i}'(b) = s_{i+1}'(b)$.
    By the induction hypothesis, we know that $s_{i}'(a) \neq s_{i}'(b)$, from which follows that $s_{i+1}'(a) \neq s_{i+1}'(b)$. This is a contradiction.

    Therefore, no two agents occupy the same vertex during the same turn. Also, notice that the wanted properties between $s_{i+1}$ and $s_{i+1}'$ hold by the construction. Finally, since $s_\ell=t$, we conclude that $s'_\ell=t$.

    This shows that any feasible solution of $\langle G,A, s_0 , t, \ell \rangle$ can be transformed into a feasible solution of $\langle H, A, s_0 , t, \ell\rangle$. Therefore, $\langle H, A, s_0 , t, \ell \rangle$ is a kernel of the claimed size.
\end{proof}

\subsection{Bounded number of agents plus diameter}

\begin{theorem}\label{thm:fpt-agents-diameter}
	The \MAPFShort problem can be solved by an FPT-algorithm parameterised by the diameter~$d$ of~$G$ plus the number of agents~$k$ in $2^{O(k^2 \cdot d \cdot \log d)} \cdot m \cdot \log n$ time.
\end{theorem}

The result is a consequence of the fact that in any sufficiently large graph $G$, there is always a feasible swap-free solution with makespan at most $O(d \cdot k)$.

\begin{lemma}\label{lem:diameter-makespan}
Let $\langle G, A, s_0, t, \ell\rangle$ be an instance of \MAPFShort such that for every $a \in A$, the vertex $t(a)$ is reachable from $s_0(a)$ and $G$ has at least $(5 \cdot d \cdot k)^{d+1}$ vertices where $d$ is the diameter of~$G$ and $k$ is the number of agents.
There exists a swap-free feasible solution of makespan $O(d \cdot k)$.
\end{lemma}
\begin{proof}
Let us assume that $G$ is connected as otherwise, the claim follows by considering each connected component separately.
Any graph with maximum degree $\Delta>2$ and diameter~$d$ has its number of vertices bounded by the Moore bound $1 + \Delta \frac{(\Delta-1)^d-1}{\Delta - 2}$, see~\cite{HoffmanS60}.
Therefore, there is a vertex $v$ of degree at least $5 \cdot d \cdot k$ in $G$.

Without loss of generality, let $A = [k]$ be the set of agents.
For each agent $i \in A$, let $P_i$ be an arbitrary shortest path between $s_0(i)$ and $t(i)$ and let $H$ be the subgraph of $G$ obtained as the union $\bigcup_{i \in [k]} P_i$.
First, observe that $H$ contains at most $k \cdot (d+1)$ vertices since each $P_i$ is a shortest path in $G$.
If $H$ is disconnected, it consists of at most $k$ connected components and we keep connecting disconnected components using shortest paths until we make it connected.
Thereby, we obtain a connected graph on at most $2 \cdot k \cdot (d+1)$ vertices including the starting and target vertices of all agents.

Next, we construct a graph $H'$ from $H$ in the following way.
If $H$ does not contain the high-degree vertex $v$, we connect it by adding an arbitrary shortest path between $v$ and $H$.
In doing so, we increase the number of vertices by at most $d-1$.
Moreover, we add to $H'$ arbitrary $k$ neighbors $w_1, \dots, w_k$ of~$v$ that are contained neither in $H$ nor in the added shortest path between $H$ and $v$.
This is possible since there are at most $2 \cdot k \cdot (d+1) + d-1 \le 4 \cdot d\cdot k$ such vertices and $v$ has at least $5 \cdot d \cdot k$ neighbors.
In total, $H'$ contains at most $O(d \cdot k)$ vertices.

Let $T$ be an arbitrary spanning tree of $H'$ and let $i_1,\dots, i_k$ be an ordering of the set of agents $[k]$ such that $\dst_T(s_0(i_j), v) \le \dst_T(s_0(i_{j+1}),v)$ for every $j \in [k-1]$ where $\dst_T(x,y)$ is the distance between $x$ and $y$ in $T$.
In other words, we order agents in increasing order with respect to the distance between their starting positions and $v$ in $T$.

We construct a feasible solution with short makespan in the following way.
We choose to describe the movements of all agents instead of tediously defining all the functions $s_1, s_2, \ldots$.
For $j \in [k]$, the agent $i_j$ starts moving at time~$j$ and follows the shortest path from $s_0(i_j)$ to $w_{i_j}$ in~$T$ without any further delay.
We claim that there cannot be any conflicts.
Assume for a contradiction that there exist $j, j' \in [k]$ such that the agents $i_j$ and $i_{j'}$ collide at time~$p$.
Let us assume that $j < j'$.
Observe that such a collision must occur only once the agents start moving and before they reach the vertex $v$ because of the ordering by distances from $v$.
In particular the agent $i_j$ is at time $p$ in distance $\dst_T(s_0(i_j),v) - p + j - 1$ from vertex $v$ and the agent $i_{j'}$ is at distance $\dst_T(s_0(i_{j'}),v) - p + j' - 1$.
However, we have $\dst_T(s_0(i_j),v) \le \dst_T(s_0(i_{j'}),v)$ due to the way we defined the ordering and $j < j'$ by assumption.
Therefore, the distance of the agent $i_j$ is strictly smaller than that of agent $i_{j'}$ and we reach contradiction.

Once we stashed every agent $i \in [k]$ in the leaf $w_i$, we use the same strategy in reverse with respect to their target positions.
In the first half of the process, it takes $k$ turns before all the agents start moving and afterwards, they must all arrive at the leaves $w_1, \dots, w_k$ in at most $O(d \cdot k)$ turns since that is the total number of vertices in~$T$.
The same bound holds for moving the agents from $w_1, \dots, w_k$ to their target positions and thus, the total makespan of this solution is $O(d \cdot k)$.
\end{proof}

\begin{proof}[Proof of Theorem~\ref{thm:fpt-agents-diameter}]
Let $n$ denote the number of vertices of $G$.
If $n < (5\cdot d \cdot k)^{d+1}$, we simply search the $k$-agent search space of size $(d \cdot k)^{O(kd)}$ using any of the standard graph searching algorithms (BFS, A*, etc.).
Otherwise, we distinguish two cases.
Either $\ell < C \cdot d \cdot k$ where the constant~$C$ is given by Lemma~\ref{lem:diameter-makespan}, and we invoke the FPT-algorithm of Theorem~\ref{thm:fpt-agents-makespan} on the given instance.
Otherwise, we invoke the same FPT-algorithm on a modified instance with $\ell = C \cdot d \cdot k$ which is guaranteed to be equivalent by Lemma~\ref{lem:diameter-makespan}.
\end{proof}

\begin{corollary}
Let $\mathcal{G}$ be a class of bounded treedepth, modularwidth, shrubdepth, or MIM-width.
The \MAPFShort{} problem is in FPT parameterised by the number of agents $k$ if $G \in \mathcal{G}$.
\end{corollary}

\begin{proof}
In each case, the diameter of $G$ is bounded by some function of the respective parameter and thus, the result follows directly from Theorem~\ref{thm:fpt-agents-diameter}.
In fact, we show that in each case the respective parameter implies non-existence of long (induced) paths.
A graph with treedepth $w$ cannot contain path of length $2^{w}$ even as a subgraph, see~\cite{NesetrilOM12}.
Any graph class with bounded shrubdepth cannot contain arbitrarily long induced paths, see~\cite{GanianHNOM19}.
Modularwidth is monotone under taking induced subgraphs, and a path of $n$ vertices has modularwidth $n$.
Therefore, a graph with modularwidth $k$ cannot contain an induced path of length $k$.
Finally, the MIM-width of a graph~$G$ is defined as the size of the largest induced matching in~$G$.
Thus clearly, there cannot be an induced path of length more than $2w+1$ in any graph with MIM-width~$w$.
\end{proof}

\section{Hardness}

\subsection{Number of agents}

Allow us to first remind the following notation. In a tree $T$ rooted at $r$, the $depth(v)$ is the number of edges in a shortest path from $v$ to $r$. The $height(T)$ is $\max_{v\in V(T)}depth(v)+1$.

\thmwhardnumberagentsDelta*

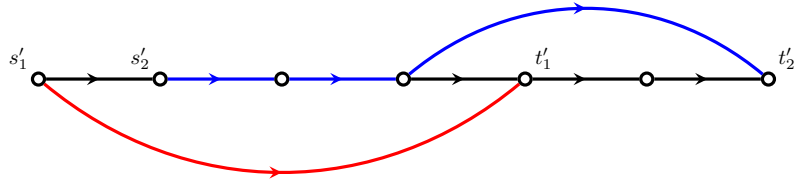
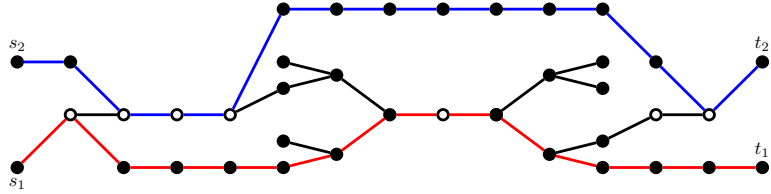
\begin{figure}[!t]
\centering

\subfloat[A topological ordering of a directed acyclic graph $D$.]{
\scalebox{0.8}{
\begin{tikzpicture}[inner sep=0.7mm]

\node[draw,black,circle,line width=1.5pt,fill=white] (v1) at (0,0)[label=above left:$s'_1$]{};
\node[draw,black,circle,line width=1.5pt,fill=white] (v2) at (2,0)[label=above left:$s'_2$]{};
\node[draw,black,circle,line width=1.5pt,fill=white] (v3) at (4,0){};
\node[draw,black,circle,line width=1.5pt,fill=white] (v4) at (6,0){};
\node[draw,black,circle,line width=1.5pt,fill=white] (v5) at (8,0)[label=above right:$t'_1$]{};
\node[draw,black,circle,line width=1.5pt,fill=white] (v6) at (10,0){};
\node[draw,black,circle,line width=1.5pt,fill=white] (v7) at (12,0)[label=above right:$t'_2$]{};

\draw[->-,line width=1.5pt,black] (v1) -- (v2);
\draw[->-,line width=1.5pt,blue] (v2) -- (v3);
\draw[->-,line width=1.5pt,blue] (v3) -- (v4);
\draw[->-,line width=1.5pt,black] (v4) -- (v5);
\draw[->-,line width=1.5pt,black] (v5) -- (v6);
\draw[->-,line width=1.5pt,black] (v6) -- (v7);

\draw[->-,line width=1.5pt,red] (v1) to[out=-40,in=-140] (v5);
\draw[->-,line width=1.5pt,blue] (v4) to[out=40,in=140] (v7);

\end{tikzpicture}
}
}

\subfloat[The graph $G$ constructed based on $D$.]{
\scalebox{0.7}{
\begin{tikzpicture}[inner sep=0.7mm]

\node[draw,black,circle,line width=1.5pt,fill=black] (v1) at (0,1)[label=above:$s_2$]{};
\node[draw,black,circle,line width=1.5pt,fill=black] (v2) at (0,-1)[label=below:$s_1$]{};
\node[draw,black,circle,line width=1.5pt,fill=white] (v3) at (1,0){};
\node[draw,black,circle,line width=1.5pt,fill=white] (v4) at (2,0){};
\node[draw,black,circle,line width=1.5pt,fill=white] (v5) at (3,0){};
\node[draw,black,circle,line width=1.5pt,fill=white] (v6) at (4,0){};

\node[draw,black,circle,line width=1.5pt,fill=black] (v7) at (5,-1){};
\node[draw,black,circle,line width=1.5pt,fill=black] (v8) at (5,-0.5){};
\node[draw,black,circle,line width=1.5pt,fill=black] (v9) at (5,0.5){};
\node[draw,black,circle,line width=1.5pt,fill=black] (v10) at (5,1){};
\node[draw,black,circle,line width=1.5pt,fill=black] (v11) at (6,-0.75){};
\node[draw,black,circle,line width=1.5pt,fill=black] (v12) at (6,0.75){};
\node[draw,black,circle,line width=1.5pt,fill=black] (v13) at (7,0){};

\node[draw,black,circle,line width=1.5pt,fill=white] (v14) at (8,0){};

\node[draw,black,circle,line width=1.5pt,fill=white] (v15) at (9,0){};
\node[draw,black,circle,line width=1.5pt,fill=black] (v16) at (9,0){};
\node[draw,black,circle,line width=1.5pt,fill=black] (v17) at (10,-0.75){};
\node[draw,black,circle,line width=1.5pt,fill=black] (v18) at (10,0.75){};
\node[draw,black,circle,line width=1.5pt,fill=black] (v19) at (11,-1){};
\node[draw,black,circle,line width=1.5pt,fill=black] (v20) at (11,-0.5){};
\node[draw,black,circle,line width=1.5pt,fill=black] (v21) at (11,0.5){};
\node[draw,black,circle,line width=1.5pt,fill=black] (v22) at (11,1){};

\node[draw,black,circle,line width=1.5pt,fill=white] (v23) at (12,0){};
\node[draw,black,circle,line width=1.5pt,fill=white] (v24) at (13,0){};
\node[draw,black,circle,line width=1.5pt,fill=black] (v25) at (14,-1)[label=above:$t_1$]{};
\node[draw,black,circle,line width=1.5pt,fill=black] (v26) at (14,1)[label=above:$t_2$]{};

\draw[line width=1.5pt,red] (v2) -- (v3);
\draw[line width=1.5pt,black] (v3) -- (v4);
\draw[line width=1.5pt,blue] (v4) -- (v5);
\draw[line width=1.5pt,blue] (v5) -- (v6);
\draw[line width=1.5pt,black] (v6) -- (v9);

\draw[line width=1.5pt,black] (v10) -- (v12);
\draw[line width=1.5pt,black] (v9) -- (v12);
\draw[line width=1.5pt,black] (v8) -- (v11);
\draw[line width=1.5pt,red] (v7) -- (v11);
\draw[line width=1.5pt,red] (v11) -- (v13);
\draw[line width=1.5pt,black] (v12) -- (v13);

\draw[line width=1.5pt,red] (v13) -- (v14);
\draw[line width=1.5pt,red] (v14) -- (v15);

\draw[line width=1.5pt,red] (v16) -- (v17);
\draw[line width=1.5pt,black] (v16) -- (v18);
\draw[line width=1.5pt,red] (v17) -- (v19);
\draw[line width=1.5pt,black] (v17) -- (v20);
\draw[line width=1.5pt,black] (v18) -- (v21);
\draw[line width=1.5pt,black] (v18) -- (v22);

\draw[line width=1.5pt,black] (v20) -- (v23);
\draw[line width=1.5pt,black] (v23) -- (v24);
\draw[line width=1.5pt,blue] (v24) -- (v26);

\node[draw,black,circle,line width=1.5pt,fill=black] (u1) at (1,1){};
\draw[line width=1.5pt,blue] (v1) -- (u1);
\draw[line width=1.5pt,blue] (u1) -- (v4);

\node[draw,black,circle,line width=1.5pt,fill=black] (u2) at (2,-1){};
\node[draw,black,circle,line width=1.5pt,fill=black] (u3) at (3,-1){};
\node[draw,black,circle,line width=1.5pt,fill=black] (u4) at (4,-1){};
\draw[line width=1.5pt,red] (v3) -- (u2);
\draw[line width=1.5pt,red] (u2) -- (u3);
\draw[line width=1.5pt,red] (u3) -- (u4);
\draw[line width=1.5pt,red] (u4) -- (v7);

\node[draw,black,circle,line width=1.5pt,fill=black] (u5) at (5,2){};
\node[draw,black,circle,line width=1.5pt,fill=black] (u6) at (6,2){};
\node[draw,black,circle,line width=1.5pt,fill=black] (u7) at (7,2){};
\node[draw,black,circle,line width=1.5pt,fill=black] (u8) at (8,2){};
\node[draw,black,circle,line width=1.5pt,fill=black] (u9) at (9,2){};
\node[draw,black,circle,line width=1.5pt,fill=black] (u10) at (10,2){};
\node[draw,black,circle,line width=1.5pt,fill=black] (u11) at (11,2){};
\node[draw,black,circle,line width=1.5pt,fill=black] (u12) at (12,1){};
\draw[line width=1.5pt,blue] (v6) -- (u5);
\draw[line width=1.5pt,blue] (u5) -- (u6);
\draw[line width=1.5pt,blue] (u6) -- (u7);
\draw[line width=1.5pt,blue] (u7) -- (u8);
\draw[line width=1.5pt,blue] (u8) -- (u9);
\draw[line width=1.5pt,blue] (u9) -- (u10);
\draw[line width=1.5pt,blue] (u10) -- (u11);
\draw[line width=1.5pt,blue] (u11) -- (u12);
\draw[line width=1.5pt,blue] (u12) -- (v24);

\node[draw,black,circle,line width=1.5pt,fill=black] (u13) at (12,-1){};
\node[draw,black,circle,line width=1.5pt,fill=black] (u14) at (13,-1){};
\draw[line width=1.5pt,red] (v19) -- (u13);
\draw[line width=1.5pt,red] (u13) -- (u14);
\draw[line width=1.5pt,red] (u14) -- (v25);

\end{tikzpicture}
}
}

\caption{An example for the construction used in the proof of Theorem~\ref{thm:w-hardness-number-agents}. The color white denotes the original vertices, while the auxiliary vertices are colored black. The colors red and blue denote the two vertex-disjoint paths of $D$ in subfigure (a), and the corresponding trajectories of $G$ that the two agents will follow in subfigure (b). Vertices lying on the same vertical line in (b) belong in the same (sub-)layer.}
\label{figure:example-w-hardness-agents}
\end{figure}

\begin{proof}
The reduction is from the \textsc{$k$-disjoint shortest paths} problem, which was shown to be $\W[1]$-hard parameterised by $k$ in~\cite{lo:21}, even when the graph given in the input is directed and acyclic. The input of this problem consists of a graph $G$ and a set of pairs of vertices $\mathcal{P}=\{(s_i,t_i):i\in [k]\}$; the question is whether there exists a set of $(s_i,t_i)$-paths hat are of minimum length and pairwise vertex-disjoint.

Let $\langle D,\mathcal{P}'\rangle$ be an instance of the \textsc{$k$-disjoint shortest paths} problem, where $D$ is a directed acyclic graph.
Let $v_1,\dots, v_n$ be a topological ordering of the vertices of $D$. We define \emph{layers} $l_0,\dots,l_{n+1}$ such that $v_i\in l_i$ for all $i\in[n]$ while $l_0$ and $l_{n+1}$ are initially empty. We begin constructing the graph for our reduction, illustrated in Figure~\ref{figure:example-w-hardness-agents}. We start with $G$ being the underlying (undirected) graph of $D$. Let us say that the vertices of $G$ at this stage are the \emph{original vertices}. Let $A=\{a_i\mid i\in[k]\}$ denote the set of agents. For each $i\in[k]$, add the vertices $s_0(a_i)$ and $t(a_i)$ in layers $l_0$ and $l_{n+1}$ respectively, as well as the edges $s_0(a_i)s'_i$ and $t(a_i)t'_i$. Observe that the vertices we added to the layers $l_0$ and $l_{n+1}$ are all of degree $1$. Currently, $G$ contains a unique vertex in $l_0$ for each $s'_i$, and a unique vertex in layer $l_{n+1}$ for each $t'_i$, $i\in [k]$, and is of maximum degree $\Delta\leq n$. In what follows, we modify $G$ so that $\Delta(G)\leq 3$ and all the shortest $(s_0(a_i),t(a_i))$-paths are of the same length.

We now modify $G$ so that $\Delta(G)\leq 3$ as follows. Let $v_i$, $i\in[n]$, be the first vertex such that $d_G(v_i)\geq 4$. First, add to $G$ the trees $T_i^1$ and $T_i^2$, both being a copy of the complete binary tree of height $H=\lceil \log(\Delta(D)+1)\rceil+1$. Let $r_i^1$ and $r_i^2$ denote the roots of $T_i^1$ and $T_i^2$ respectively. Then, add the edges $v_ir_i^1$ and $v_ir_i^2$. We insert the vertices of the newly added trees in new \emph{sub-layers} of $l_i$ defined as follows. For each vertex $u$ that belongs to $T_i^1$ ($T_i^2$ resp.) and $u$ is at depth $0\leq h\leq H-1$, add $u$ inside the layer $l_{i,h}^1$ ($l_{i,h}^2$ resp.). Then, for each vertex $v$ that belongs in a layer $l_j$ or a sub-layer $l_{j,H}^2$, for $j<i$ (or a sub-layer $l_{j,H}^1$ for $j>i$ resp.), such that $vv_i\in E(G)$, delete the edge $vv_i$ and add an edge between $v$ and any one of the vertices of $T_i^1$ ($T_i^2$ resp.) that is still of degree $1$ (whose existence is guaranteed by the height of the the binary trees we used). Repeat this process as long as there exist vertices in $G$ that are of degree greater than $3$. This terminates after at most $n$ repetitions, and adds in total at most $\mathcal{O}(n^2)$ vertices to $G$. The resulting graph $G$ has maximum degree at most $3$.

Lastly, we modify $G$ so that all the shortest $(s_0(a_i),t(a_i))$-paths are of the same length. To do so, we will make use of the layers and sub-layers defined thus far. Let $uv\in E(G)$ and $i,j$ be such that $u$ and $v$ belong in layers $l$ and $l'$, respectively, with $l\in \{l_i,l_{i,H}^1,l_{i,H}^2\}$ and $l'\in \{l_j,l_{j,H}^1,l_{j,H}^2\}$ and $|i-j|>1$; such edges are called \emph{bad}. We subdivide each bad edge $e\in E(G)$ so that each layer and sub-layer between $l$ and $l'$ receives a new vertex. Observe that in the resulting graph, all the shortest $(s_0(a_i),t(a_i))$-paths have the same length $L$
Finally, all the vertices of $G$ at this stage that are not original, are denoted as \emph{auxiliary}. This finishes the construction of $G$.

We are now ready to proceed with our reduction. In particular, we will show that $\ell(\langle G,A,s_0,t\rangle)\leq L$ if and only if $\langle D,\mathcal{P}'\rangle$ is a yes-instance of the \textsc{$k$-disjoint shortest paths} problem.

Let us first assume that $\langle D,\mathcal{P}'\rangle$ is a yes-instance of the \textsc{$k$-disjoint shortest paths} problem and let $\{P'_i=(s_i'=u_1^i,\dots,u_p^i=t_i'):i\in[k],p\leq n\}$ be a set of disjoint shortest $(s_i',t_i')$-paths in $D$. For each $i\in[k]$ we define an $(s_0(a_i),t(a_i)$ path $P_i$ as follows. The first and last vertices of $P_i$ are $s_0(a_i)$ and $t(a_i)$ respectively. Then, starting from $s_0(a_i)$, we follow the edges of $G$ until we reach an original vertex $v$. From the construction of $G$, it follows that $v=s_i'$. We proceed by following the original vertices of $G$ according to $P'_i$ as long as possible. Let us assume that it is not possible to do so, and let $u$ be the last (original) vertex we have inserted into $P_i$ and $v$ be the next vertex proposed by $P'_i$.

\begin{claim}
If $(uv)\in A(D)$ and $uv\notin E(G)$, there exists a unique $(u,v)$-path $P$ in $G$.
\end{claim}
\begin{proofclaim}
Let $u\in l_i$ and $v\in l_j$. If $d_D(u)\leq d_D(v)\leq 3$, then it follows that during the construction of $G$, the edge $uv$ was a bad edge that got subdivided into a path, which is exactly $P$. Otherwise, at least one of $T_i^2$ and $T_j^1$ exists in $G$. We distinguish cases:
\begin{itemize}
 \item if both $T_i^2$ and $T_j^1$ exist, then by the construction of $G$ we have that there exist unique vertices $w\in l_{i,H}^2$ and $z\in l_{j,H}^1$ such that $wz\in E(G)$ or there exists a unique $(w,z)$-path (the result of the subdivision of a bad edge). Follow the unique path leading from $u$ to $w$, continue to $z$, and then follow the unique path leading from $w$ to $v$;
 \item if only $T_i^2$ exists, then by the construction of $G$ we have that there exists a unique vertex $w\in l_{i,H}^2$ such that $wv\in E(G)$ or there exists a unique $(w,v)$-path (the result of the subdivision of a bad edge). Follow the unique path leading from $u$ to $w$ and then continue to $v$;
 \item if only $T_j^1$ exists, then by the construction of $G$ we have that there exists a unique vertex $z\in l_{j,H}^1$ such that $uz\in E(G)$ or there exists a unique $(u,z)$-path (the result of the subdivision of a bad edge). Start by moving from $u$ to $z$ and then follow the unique path leading from $z$ to $v$.
\end{itemize}
\end{proofclaim}

So, following the unique path as it is defined in the above claim, we are able to reach $v$. From there, we continue following $P'_i$ as long as possible.
This finishes the definition of $P_i$. Observe that, for each $i\in[k]$, we have that $P_i$ is of length exactly $L$, as it contains exactly one vertex of each layer and sub-layer of $G$. Also, observe that by their construction, the paths $P_1,\dots,P_k$ are pairwise vertex-disjoint paths. This means that any two agents $a_i$ and $a_j$ that, at each turn, move from their current vertex to the next one according to $P_i$ and $P_j$ respectively, are never going to attempt to go on the same vertex, and they are going to reach their destinations after $L$ turns. Thus, $\ell(\langle G,A,s_0,t\rangle)\leq L$.

For the reverse direction, assume that $\ell(\langle G,A,s_0,t\rangle)= L$ and let $s_i$, $i\in[L]$, be the sequence of functions such that $s_i:A\rightarrow V(G)$ gives us the set of vertices on which the agents are placed at the end of the $i^{th}$ turn. For each $j\in[k]$, let the path $P'_j$ be defined as follows. Consider the turns $i_1,\dots,i_p$ (for $p\leq L$) such that the agent $a_j$ is visiting an original vertex of $G$ according to $s_{i_m}$, $m\in [p]$. Let $P'_j=(s_{i_1}(a_j),\dots, s_{i_p}(a_j))$. It follows from the construction of $G$ that $P'_j$ is indeed a path of $D$
and that $s_{i_1}(a_j)=s_j'$ and $s_{i_p}(a_j)=t_j'$. Also, since all the $(s_0(a_i),t(a_i))$-paths of $G$ are of length $L$, we have that at each turn, all the agents are located on vertices of the same layer or sub-layer of $G$. Moreover, by the construction of $G$, each layer contains at most one original vertex. It follows that $P'_j$ contains at most one vertex for each layer of $D$ and, thus, that $P'_j$ is a shortest $(s_j',t_j')$-path. It remains to show that these paths that we defined are also pairwise vertex-disjoint. Towards a contradiction, assume that there are $j_1,j_2\in [k]$ such that the paths $P'_{j_1}$ and $P'_{j_2}$ intersect on a vertex $v$. Then, by the definition of these paths, we have that there exists a turn $i$ such that $s_{i}(a_{j_1})=s_{i}(a_{j_2})=v$, which  contradicts the feasibility of the solution of the \textsc{Multiagent Pathfinding Problem}.
\end{proof}

\subsection{Trees of bounded maximum degree}

\thmMAPFisNPCTreesDelta*

\begin{proof}
    It was recently shown that the \textsc{Parallel Token Swapping} is NP-hard on trees~\cite{ADKLLMRWW22}. This problem is exactly the same as the \textsc{Multiagent Pathfinding Problem} when swaps are allowed and the number of agents is equal to the number of vertices of the tree.
    Here, we present a reduction from this version of the problem to the version where swaps are not allowed, the input graph is a tree with maximum degree $5$, and the number of agents is arbitrary.
    Let $\langle T',A', \turn_0' , t'\rangle$ be an instance of \textsc{Multiagent Pathfinding Problem}, where swaps are allowed and $T'$ is a tree.
    We will construct an instance $\langle T, A, \turn_0 , t \rangle$ of the  \textsc{Multiagent Pathfinding Problem} where swaps are not allowed and $T$ is a tree of maximum degree $5$, such that $\ell(\langle T', A', \turn_0' , t' \rangle)\leq \ell'$ if and only if $\ell(\langle T, A, \turn_0 , t \rangle)\leq \ell$, for a given value $\ell'$.

    First we present the construction of $T=(V,E)$. Start with $T'=(V',E')$ and let $h$ be the value $ \{ \lceil \log \Delta (T') \rceil +1 \}$ if this is even, otherwise, $h= \{ \lceil \log \Delta (T') \rceil +2 \}$.
    We will bound the degree of this graph by replacing the edges between any non-leaf vertex and its children with a
    complete binary tree of height $h$.
    In particular, for each non-leaf vertex $v \in V'$,
    we create $T_v$, a complete binary tree of height $ h $ rooted by a vertex $r_v$.
    Let $U$ be the set $\bigcup_{v\in V'} V(T_v)$ and $E'= \bigcup_{v\in V'} E(T_v)$. We also add to $E'$ the set of
    edges $\{vr_v: v \in V'\}$.
    Now we will replace the edges of the original graph.
    For each non-leaf vertex $v \in V'$ let $v_i$, $i \in [d_{T'}(v)]$, be the children of $v$ in $T'$ and
    $l_i$, $i \in [d_{T'}(v)]$, be $d_{T'}(v)$ distinguished leaves in $T_v$.
    We remove $\{ vv_i : i \in [d_{T'}(v)] \}$ from $E'$ and we add $\{ v_il_i : i \in [d_{T'}(v)] \}$ to $E'$.
    Notice that the resulting graph 
    is a tree of maximum degree $3$.

    To complete the construction, for each vertex $u \in U$, we add two paths $P^{u,s} = ( u_{s,0},\ldots u_{s,\ell' m -1} )$
    and $P^{u,t} = ( u_{t,1},\ldots u_{t,(\ell'-1)m} )$, where $m = 2 + h$,
    and the edges $uu_{s,\ell' m -1}$ and $uu_{t,1}$. Let $W$ be the set of all the vertices in these paths and
    $E_W$ all the edges with at least one incident vertex belonging in $W$.
    The final graph is $T = (V'\cup U\cup W, E'\cup E_W)$.
    Notice that the addition of $E_W$ increases the maximum degree of $T$ to at most $5$.

    Next, we define the set of agents $A$ and the functions $\turn_0$ and $t$.
    We start by setting $A=A'$; we will call these the \textit{original agents}. Also for each vertex $u \in U$ we will add $\ell'$ \textit{path agents} $a_{{u,i}}$, $i \in [\ell']$ to $A$.
    We define the function $\turn_0$ as follows. For any original agent $a$, $\turn_0 (a) = \turn_0'(a)$.
    For each path agent $a_{{u,i}}$, where $u \in U$ and $i \in [\ell']$, $\turn_0 (a_{{u,i}}) = u_{s,(i-1)m}$.
    Finally, we define the function $t$ as follows. For any original agent $a$, $t (a) = t'(a)$.
    For each path agent $a_{{u,i}}$, where $u \in U$ and $i \in [\ell']$, if $i>1$ then $t (a_{{u,i}}) = u_{t,(i-1)m}$ and $t (a_{{u,1}})  = u$.

    \begin{claim}
        If $\ell(\langle T', A', \turn_0' , t' \rangle)\leq \ell'$ then $\ell(\langle T, A, \turn_0 , t \rangle)\leq\ell = \ell' m$.
    \end{claim}
    \begin{proofclaim}

    Let $\turn_0',\turn_1',\ldots, \turn_{\ell'}'$ be a feasible solution of $\langle T', A', \turn_0' , t' \rangle$.
    We will define a new sequence $\turn_0,\turn_1,\ldots, \turn_{\ell}$ that is a feasible solution of
    $\langle T, A, \turn_0 , t \rangle$. To do so, for each $\turn_i'$, $i \in [\ell']$, we create a sequence $\turn_{(i-1)m+1},\ldots,\turn_{im}$. The definition is inductive. First, $\turn_0$ is the starting position of each agent in $\langle T,A,\turn_0,t\rangle$. Assuming that we have defined the sequence until the $\turn_{(i-1)m}$, we define the $\turn_{(i-1)m+1},\ldots,\turn_{im}$ as follows:

    \begin{itemize}
        \item for each agent $a_{u,p}$, $u \in U$ and $p \in [\ell']$, for $j=1$ to $m$:
        \begin{itemize}
            \item if $\turn_{(i-1)m+j-1}(a_{u,p}) = u$ then $\turn_{(i-1)m+j}(a_{u,p}) = u_{t,1}$;
            \item if $\turn_{(i-1)m+j-1}(a_{u,p}) = u_{s, \ell'm-1}$ then $\turn_{(i-1)m+j}(a_{u,p}) = u$;
            \item if $\turn_{(i-1)m+j-1}(a_{u,p}) = u_{s, q}$ for some $q \in \{ 0\} \cup [\ell'm-2]$, then $\turn_{(i-1)m+j}(a_{u,p}) = u_{s,q+1}$;
            \item if $\turn_{(i-1)m+j-1}(a_{u,p}) = u_{t, q}$ for some $ q \in  [(i-1)m-1] $, then $\turn_{(i-1)m+j}(a_{u,p}) = u_{t, q+1}$.
        \end{itemize}
    \end{itemize}
    Now, we need to deal with the original agents. Notice that, since $T'$ is acyclic, for any agent $a$ and turn $i \in [\ell']$, we either have that $\turn_{i-1}'(a)= \turn_{i}'(a)$ or there exists an agent $b$ such that $\turn_{i-1}'(a)= \turn_{i}'(b)$ and $\turn_{i-1}'(b)= \turn_{i}'(a)$ ($a$ and $b$ swap positions).

    \begin{itemize}
        \item for each original agent $a$, such that $\turn_{i-1}'(a)= \turn_{i}'(a)$:
        \begin{itemize}
            \item we create the sequence $\turn_{(i-1)m+1}(a) = \ldots = \turn_{im}(a) = \turn_{i}'(a)$.
        \end{itemize}
        \item for each pair of original agents $(a, b)$, such that $\turn_{i-1}'(a)= \turn_{i}'(b)$ and $\turn_{i-1}'(b)= \turn_{i}'(a)$:
        \begin{itemize}
            \item W.l.o.g., assume that $\turn_{i-1}'(a)$ is the parent of $\turn_{i}'(a)$ in $T'$. Let $(\turn_{i-1}'(a)=v_0, v_1, \ldots, v_h, v_{h+1}= \turn_{i}'(a))$ be the unique path between $\turn_{i-1}'(a)$ and $\turn_{i}'(a)$ in $T$ and $v'_{\nicefrac{h}{2}+1}$ be the child of $v_{\nicefrac{h}{2}}$ that does not belong in that path.
            \item For $b$ we set $\turn_{im-1}(b) = \turn_{im}(b)= \turn'_{i}(b) $ and $\turn_{(i-1)m+j}(b) = v_{h-j+1}$ for all $j\in [h]$.
            \item For $a$ we set $\turn_{(i-1)m+j}(a) = v_{j}$ for all $j\in [\nicefrac{h}{2}]$, $\turn_{(i-1)m+\nicefrac{h}{2}+1}(a)= v'_{\nicefrac{h}{2}+1}$ and $\turn_{(i-1)m+j}(a) = v_{j-1}$ for all $j\in [h+2]\setminus [\nicefrac{h}{2}+1]$.
        \end{itemize}
    \end{itemize}
    Notice that the previous sequence is feasible as the agents always move form one vertex to one of its neighbors and no two agents occupy the same vertex at the end of the same turn.
    \end{proofclaim}

    It remains to show that if $\ell(\langle T, A, \turn_0 , t \rangle)\leq\ell$ then,
    $\ell(\langle T', A', \turn_0' , t' \rangle)\leq \ell'$.
    Assume that $\ell(\langle T, A, \turn_0 , t \rangle)=\ell$ and let $\turn_1,\ldots, \turn_\ell$ be a feasible solution.
    First we show that, in any such solution, the movements of path agents are known.
    \begin{claim}\label{claim:para-np-hard-trees-path-agents}
        In any solution of $\langle T, A, \turn_0 , t \rangle$ that has makespan $\ell=m\ell'$, any path agent $a$ needs to constantly move through the shortest path between its starting and finishing positions. Also, any vertex $u \in U$ is occupied by the agent $a_{u,i}$ during the turn $im$.
    \end{claim}
    \begin{proofclaim}
    Let $a_{u,i}$, for some $u \in U$ and $i \in \ell'$, be a path agent. The unique path between $s_{u,i}$ and $t_{u,i}$ is $(u_{s,im -1},\ldots u_{s,\ell m -1},u,u_{t,1},\ldots, u_{t,im})$. This path is of length $m\ell'$ (number of edges) therefore $a_{u,i}$ needs to constantly move through this path in order for a sequence with makespan $\ell=m\ell'$ to be valid.
    For the second part of the claim, it suffices to observe that $u$ belongs in the path that $a_{u,i}$ needs to move through, and $s_{u,i} = u_{s,im -1}$ has distance exactly $im$ from $u$.
    \end{proofclaim}

    Now, we show that the original agents occupy vertices of $V'$ during specific turns.
    \begin{claim}\label{claim:para-np-hard-trees-original-agents}
        For any original agent $a$ it holds that:
        \begin{itemize}
            \item for all $i \in [\ell'] \cup \{0\}$, $\turn_{im}(a) \in V'$ and
            \item for all $i \in [\ell']$,  $\turn_{im}(a)$ is a neighbor of $\turn_{(i-1)m}(a)$ in $T'$.
        \end{itemize}
    \end{claim}
    \begin{proofclaim}
    Due to Claim~\ref{claim:para-np-hard-trees-path-agents}, we know that $\turn_{im}(a) \notin U$.
    We need to prove that that $\turn_{im}(a) \notin W$.
    Assume that there exists an original agent $a$ such that $\turn_{im}(a) = v \in W$ for an $i \in \ell'$.
    We consider two cases: either $v = u_{t,q}$ for some $u \in U$ and $q \in [(\ell' - 1)m]$ or $v = u_{s,q}$ for some $u \in U$ and $q \in [\ell'm-1]$.

    \textbf{Case 1. $\boldsymbol{(v = u_{t,q})}$:} Notice that there exists a path agent $b$ such that $\turn_{im}(b) = u$ and, in any solution of makespan $\ell$, they need to continue moving through the path $P^{u,t}$. Since swaps are not allowed, the agent $a$ cannot leave the path $P^{u,t}$. Therefore this sequence cannot be a feasible solution as $t (a) \in V$ by construction. 

    \textbf{Case 2. $\boldsymbol{(v = u_{s,q})}$:} We show that this is impossible to happen. Let $i$ be the minimum integer such that $\turn_{im}(a) = u_{t,q}$. For the path agent $b=a_{u,\ell'-i}$, by the Claim~\ref{claim:para-np-hard-trees-path-agents} and the fact that $\turn_0(b) = u_{s,(\ell'-i)m-1}$, we know that $\turn_{(i-1)m +j-1}(b) = u_{s,(\ell'-1)m +j-1}$ for all $j \in [m]$. Finally, since $i$ is the minimum integer such that $\turn_{im}(a) = u_{t,q}$, we know that $\turn_{(i-1)m}(a) \notin P^{u,s}$. Therefore, if $\turn_{im}(a) = u_{t,q}$, then $a$ and $b$ must swap positions during some turn between $(i-1)m$ and $im$, contradicting the fact that swaps are not allowed. 

    Therefore, we have shown that $\turn_{im}(a) \in V'$ for every $i\in[\ell']\cup \{0\}$. It remains to show that $\turn_{im}(a)$ is a neighbor of $\turn_{(i-1)m}(a)$ in $T'$. Observe that any two neighboring vertices in $T'$ have distance $m-2$ in $T$ while two non-neighboring vertices in $T'$ have distance at least $2m-4$ in $T$. Observe also that $h\ge 2$ and, thus, $m\ge 6$. Therefore, if $\turn_{im}(a)$ is not a neighbor of $\turn_{(i-1)m}(a)$ in $T'$, we would have that in $T$, the agent $a$ would have to travel a distance of $2m-4$ which is strictly larger than the $m$ turns during which he is supposed to do so.
    \end{proofclaim}

    Finally, we claim that the sequence $\turn'_0,\turn'_1,\ldots, \turn'_{\ell'}$, where $\turn'_i = \turn'_{im}\mid_{A'}$ (the $\turn'_{im}$ restricted to the set of the original agents), is a feasible solution of $\langle T',A',\turn_0',t'\rangle$.
    First notice that for any agent $a \in A'$, we have that $\turn'_i(a)$ is a neighbor of $\turn'_{i-1}(a)$ (by Claim~\ref{claim:para-np-hard-trees-original-agents}). Since $T'$ is a tree, we know that the sequence is feasible if for any agent $a \in A'$ and $i \in [\ell']$, where $\turn'_{i-1}(a) \neq \turn'_{i}(a)$, there exists an agent $b \in A'$ such that $\turn'_{i}(a) =  \turn'_{i-1}(b)$ and $\turn'_{i}(b) =  \turn'_{i-1}(a)$.

    Assume that for an agent $a \in A'$ and $i \in [\ell']$ we have $\turn'_{i-1}(a) \neq \turn'_{i}(a)$. Furthermore, let $b$ be the agent that occupies the vertex $\turn'_{i}(a)$ at the end of turn $i-1$.
    We will show that $\turn'_{i}(b) = \turn'_{i-1}(a)$.
    Let $V_1$ and $V_2$ be the vertex sets of the two connected components of $T' - e$, where $e = \turn'_{i-1}(a) \turn'_{i}(a)$. W.l.o.g., assume that $\turn'_{i-1}(a) \in V_1$ and $ \turn'_{i}(a) \in V_2$. Let $A_1$ and $A_2$ be the partition of the agent set $A'$ where, for any $i \in [2]$, $a' \in A_i $ implies that $\turn_{i-1}'(a') \in V_i $.
    Observe that $N(\turn'_{i-1}(a)) \cap  V_2$ contains only the vertex $\turn'_{i}(a)$. Therefore, the only agent in $A_2$ that can move to $\turn'_{i-1}(a)$ in turn $i$ is $b$ (as $b$ occupies $\turn'_{i}(a)$ in turn $i-1$). Finally, since $|A_2| = |V_2|$, if $\turn'_{i}(b) \neq \turn'_{i-1}(a)$ then, at turn $i$, we would have $|A_2| +1 $ agent occupying $|V_2|$, which is a contradiction. Therefore, the aforementioned sequence is a feasible solution of $\langle T',A',\turn_0',t\rangle$ with makespan $\ell'$. This completes the proof.
 \end{proof}

\subsection{Cliquewidth plus makespan}

\thmMAPFisWHbyCwplusMakespan*

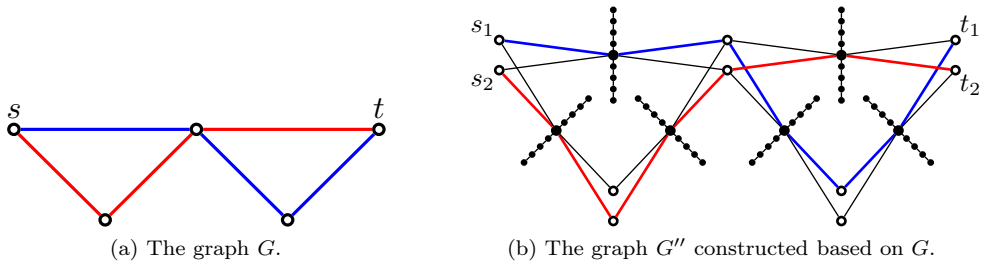
\begin{figure}[!t]
	\centering

	\subfloat[The graph $G$.]{
		\scalebox{1.2}{
			\begin{tikzpicture}[inner sep=0.4mm]

				\node[draw,black,circle,line width=1pt,fill=white] (s) at (0,1)[label=above:$s$]{};
				\node[draw,black,circle,line width=1pt,fill=white] (x) at (1,0){};
				\node[draw,black,circle,line width=1pt,fill=white] (z) at (3,0){};
				\node[draw,black,circle,line width=1pt,fill=white] (t) at (4,1)[label=above:$t$]{};
				\node[draw,black,circle,line width=1pt,fill=white] (y) at (2,1){};

				\draw[line width=1pt,blue] (s) -- (y);
				\draw[line width=1pt,red] (s) -- (x);
				\draw[line width=1pt,red] (x) -- (y);
				\draw[line width=1pt,red] (y) -- (t);
				\draw[line width=1pt,blue] (y) -- (z);
				\draw[line width=1pt,blue] (z) -- (t);
			\end{tikzpicture}
		}
	}\hspace{10pt}
	\subfloat[The graph $G''$ constructed based on $G$.]{
		\scalebox{1}{
			\tikzmath{\td = 0.2; \dist=0.15; \pd=\dist/sqrt(2);}
			\begin{tikzpicture}[inner sep=0mm,vertex/.style={draw,black,circle,line width=1pt,fill=white,minimum size=3pt},edge/.style={draw,black,circle,line width=1pt,fill=black,minimum size=3pt},path/.style={draw,black,circle,line width=1pt,fill=black,minimum size=1.5pt}]

				\node[vertex] (s1) at (0,2+\td)[label=above left:$s_1$]{};
				\node[vertex] (s2) at (0,2-\td)[label=below left:$s_2$]{};
				\node[vertex] (x1) at (1.5,0+\td){};
				\node[vertex] (x2) at (1.5,0-\td){};
				\node[vertex] (y1) at (3,2+\td){};
				\node[vertex] (y2) at (3,2-\td){};
				\node[vertex] (z1) at (4.5,0+\td){};
				\node[vertex] (z2) at (4.5,0-\td){};
				\node[vertex] (t1) at (6,2+\td)[label=above right:$t_1$]{};
				\node[vertex] (t2) at (6,2-\td)[label=below right:$t_2$]{};

				\node[edge] (e1) at (0.75,1){};
				\node[edge] (e2) at (2.25,1){};
				\node[edge] (e3) at (1.5,2){};
				\node[edge] (e4) at (3.75,1){};
				\node[edge] (e5) at (5.25,1){};
				\node[edge] (e6) at (4.5,2){};

				\foreach \x in {0,...,4}
				{
					\node[path] (p1a\x) at (0.75+\x*\pd,1+\x*\pd){};
					\node[path] (p1b\x) at (0.75-\x*\pd,1-\x*\pd){};
					\ifthenelse{\x > 0}{
						\pgfmathtruncatemacro{\dest}{\x - 1}
						\draw[line width=0.5pt,black] (p1a\x) -- (p1a\dest);
						\draw[line width=0.5pt,black] (p1b\x) -- (p1b\dest);
					}{}
				};

				\foreach \x in {0,...,4}
				{
					\node[path] (p2a\x) at (2.25-\x*\pd,1+\x*\pd){};
					\node[path] (p2b\x) at (2.25+\x*\pd,1-\x*\pd){};
					\ifthenelse{\x > 0}{
						\pgfmathtruncatemacro{\dest}{\x - 1}
						\draw[line width=0.5pt,black] (p2a\x) -- (p2a\dest);
						\draw[line width=0.5pt,black] (p2b\x) -- (p2b\dest);
					}{}
				};

				\foreach \x in {0,...,4}
				{
					\node[path] (p3a\x) at (1.5,2+\x*\dist){};
					\node[path] (p3b\x) at (1.5,2-\x*\dist){};
					\ifthenelse{\x > 0}{
						\pgfmathtruncatemacro{\dest}{\x - 1}
						\draw[line width=0.5pt,black] (p3a\x) -- (p3a\dest);
						\draw[line width=0.5pt,black] (p3b\x) -- (p3b\dest);
					}{}
				};

				\foreach \x in {0,...,4}
				{
					\node[path] (p4a\x) at (4.5,2+\x*\dist){};
					\node[path] (p4b\x) at (4.5,2-\x*\dist){};
					\ifthenelse{\x > 0}{
						\pgfmathtruncatemacro{\dest}{\x - 1}
						\draw[line width=0.5pt,black] (p4a\x) -- (p4a\dest);
						\draw[line width=0.5pt,black] (p4b\x) -- (p4b\dest);
					}{}
				};

				\foreach \x in {0,...,4}
				{
					\node[path] (p5a\x) at (3.75+\x*\pd,1+\x*\pd){};
					\node[path] (p5b\x) at (3.75-\x*\pd,1-\x*\pd){};
					\ifthenelse{\x > 0}{
						\pgfmathtruncatemacro{\dest}{\x - 1}
						\draw[line width=0.5pt,black] (p5a\x) -- (p5a\dest);
						\draw[line width=0.5pt,black] (p5b\x) -- (p5b\dest);
					}{}
				};

				\foreach \x in {0,...,4}
				{
					\node[path] (p6a\x) at (5.25-\x*\pd,1+\x*\pd){};
					\node[path] (p6b\x) at (5.25+\x*\pd,1-\x*\pd){};
					\ifthenelse{\x > 0}{
						\pgfmathtruncatemacro{\dest}{\x - 1}
						\draw[line width=0.5pt,black] (p6a\x) -- (p6a\dest);
						\draw[line width=0.5pt,black] (p6b\x) -- (p6b\dest);
					}{}
				};

				\draw[line width=1pt,blue] (s1) -- (e3);
				\draw[line width=1pt,blue] (e3) -- (y1);
				\draw[line width=1pt,blue] (y1) -- (e4);
				\draw[line width=1pt,blue] (e4) -- (z1);
				\draw[line width=1pt,blue] (z1) -- (e5);
				\draw[line width=1pt,blue] (e5) -- (t1);
				\draw[line width=1pt,red] (s2) -- (e1);
				\draw[line width=1pt,red] (e1) -- (x2);
				\draw[line width=1pt,red] (x2) -- (e2);
				\draw[line width=1pt,red] (e2) -- (y2);
				\draw[line width=1pt,red] (y2) -- (e6);
				\draw[line width=1pt,red] (e6) -- (t2);

				\draw[line width=0.5pt,black] (s2) -- (e3);
				\draw[line width=0.5pt,black] (s1) -- (e1);
				\draw[line width=0.5pt,black] (e1) -- (x1);
				\draw[line width=0.5pt,black] (x1) -- (e2);
				\draw[line width=0.5pt,black] (e2) -- (y1);
				\draw[line width=0.5pt,black] (e3) -- (y2);
				\draw[line width=0.5pt,black] (y2) -- (e4);
				\draw[line width=0.5pt,black] (y1) -- (e6);
				\draw[line width=0.5pt,black] (e4) -- (z2);
				\draw[line width=0.5pt,black] (z2) -- (e5);
				\draw[line width=0.5pt,black] (e5) -- (t2);
				\draw[line width=0.5pt,black] (e6) -- (t1);
			\end{tikzpicture}
		}
	}

	\caption{An example for the construction used in the proof of Theorem~\ref{thm:w-hardness-cw-makespan} on a yes-instance $\langle G, s,t,2,3 \rangle$ of \textsc{BEUP}.
    In subfigure (a), the colors red and blue distinguish two edge-disjoint paths of length $3$. In subfigure (b), they denote the corresponding trajectories followed by the two path agents.}
	\label{figure:example-w-hardness-cw}
\end{figure}

\begin{proof}
We reduce from the problem \textsc{Bounded Edge Undirected $(s,t)$-disjoint Paths (BEUP)} where the input consists of an undirected graph~$G$ with two distinct vertices $s,t$ and two positive integers~$k$ and~$d$; and the question is whether there are~$k$ edge-disjoint $(s,t)$-paths of length at most~$d$ in~$G$.
In~\cite{GolovachT11}, \textsc{BEUP} was shown to be $\W[1]$-hard when parameterised by the treewidth of~$G$ for every fixed $d \ge 10$.

Let $\langle G, s, t, k, d \rangle$ be an instance of \textsc{BEUP} where $G = (V,E)$.
We first describe the high level idea of the reduction.
We subdivide each edge of $G$ with a new vertex and we enforce that this vertex is blocked in all but at most one turn.
This can be done in a similar fashion to the path agents used in the proof of Theorem~\ref{thm:np-hard-trees}.
Finally, we add to $G$ many twins of each original vertex simulating that several agents can occupy a single original vertex at the same time.
Any solution to the instance $\langle G, s, t, k, d \rangle$ then corresponds to moving agents situated on $k$ twins of $s$ to arbitrary $k$ twins of $t$.

Now, we describe the construction formally.
First, we subdivide each edge $e \in E$ with a new vertex $v_e$.
Afterwards, we add a path $P_e$ on $4d-3$ vertices $v_e^1, v_e^2, \dots, v_e^{4d - 3}$ such that its middle vertex $v_e^{2d-1}$ is identified with $v_e$ and all other vertices are disjoint from the rest of the graph.
Let $G'$ be the obtained graph.
It is easy to see that $\tw(G') = \tw(G)$.
And therefore, it follows that $\cw(G') = 3 \cdot 2^{\tw(G')-1} \le 3 \cdot 2^{\tw(G)-1}$ by~\cite{CorneilR05}.

We now modify $G'$ to a graph $G''$ by replacing every original vertex $v \in V$ of $G$ with an independent set of $k$ twins $v_1, \dots, v_k$.
See Figure~\ref{figure:example-w-hardness-cw} for example.
The treewidth of $G''$ may no longer be bounded by a function of $\tw(G)$.
However, it is easy to see that $\cw(G'') = \cw(G')$ 
and thus, the cliquewidth of $G''$ is at most exponential in the treewidth of $G$.

It remains to define the set of agents $A$ together with their starting and ending positions.
The set of agents $A$ consists of a set of \emph{path agents} $A_0 = \{a_i \mid i \in [k]\}$ and a set of $2d-2$ \emph{edge agents} $B_e = \{b_e^i \mid i \in [2d-2]\}$ for every edge $e \in E$.
We set $s'_0(a_i) = s_i$, $t'(a_i) = t_i$ for every path agent $a_i \in A_0$.
In other words, the starting and ending positions of path agents lie in the independent sets corresponding to the original vertices $s$ and $t$, respectively.
For every $e \in E$ and every $i \in [2d-2]$, we set $s'_0(b_e^i) = v_e^i$ and $t'(b_e^i)=v_e^{2d + i - 1}$.
We now show that $\ell(\langle G'',A,s'_0,t'\rangle)\leq 2 d$ if and only if $\langle G, s, t, k, d \rangle$ is a yes-instance of \textsc{BEUP}.

First, let $\langle G, s, t, k, d \rangle$ be a yes-instance of \textsc{BEUP} and let $\{P'_i= (s= u^i_0, \dots, u^i_{d_i}=t) \mid i \in [k]\}$ be a set of $k$ edge-disjoint $(s,t)$-paths in $G$ such that $d_i \le d$ for every $i \in [k]$.
We first define the movement of the path agents.
We let the agent $a_i$ follow the natural subdivision of the path $P'_i$ in $G''$.
Formally for every $i \in [k]$, we set $s'_{2j}(a_i) = (u^i_j)_i$ and $s'_{2j-1}(a_i)=v_{\{u^i_{j-1}, u^i_j\}}$ for $j \in [d_i]$, and $s'_{2j-1}(a_i) = s'_{2j}(a_i) = t_i$ for $j \in [d_i+1,d]$.
All the movements of path agents are feasible since in even turns they occupy pairwise different copies of the original vertices from~$G$ and in odd turns, they occupy vertices corresponding to edges on the edge-disjoint paths $P'_1, \dots, P'_k$.

It remains to define the movement of the edge agents.
Fix an edge $e \in E$ of $G$.
Notice that the vertex $v_e$ is used by a path agent in at most one turn.
If that is the case, let $c_e$ be the time when the vertex $v_e$ is occupied by a path agent and otherwise, let $c_e = 2d$.
We define the movement of agents of $B_e$ in such a way that they make way for the path agent exactly at time $c_e$.
For every $i \in [2d - c_e: 2d]$, we let the agent $b_e^i$ walk directly to its destination using the shortest path of length $2d-1$.
In other words, we set $s'_j(b_e^i) = v_e^{j+i}$ for $j \in [2d - 1]$ and $s'_{2d}(b_e^i) = v_e^{2d+i-1} = t'(b_e^i)$.
Every other agent from $B_e$ moves towards its destination for the first $c_e - 1$ turns, then waits one turn and continues afterwards with no further delay.
Formally, we set for every $i \in [2d - c_e-1]$ and $j \in [2d]$
\[
s'_j(b_e^i) = \begin{cases}
v_e^{j+i} &\text{if $j \le c_e - 1$, and}\\
v_e^{j+i-1} &\text{if $j \ge c_e$.}
\end{cases}
\]

It is easy to check that the edge agents in $B_e$ are occupying pairwise different positions at each turn and moreover, they do not swap.

\begin{claim}
The vertex $v_e$ is at time $c_e$ not occupied by any agent of $B_e$.
\end{claim}
\begin{proofclaim}
Assume that there is an agent $b_e^i$ such that $s'_{c_e}(b_e^i) = v_e = v_e^{2d-1}$.
If $i \ge 2d -c_e$, then agent $b_e^i$ moves without any delays and $s'_{c_e}(b_e^i) = v_e^{i + c_e} \neq v_e^{2d-1}$ since $i + c_e \ge 2d$.
On the other hand, if $i < 2d - c_e$ then $s'_{c_e}(b_e^i) = v_e^{c_e + i -1} \neq v_e^{2d-1}$ since $c_e + i - 1 < 2d - 1$.
\end{proofclaim}

It follows that also edge agents and path agents occupy pairwise disjoint positions at all times.
Therefore, the sequence $s'_1, \dots, s'_{2d}$ is feasible and $\langle G'',A,s'_0,t',2d\rangle$ is a yes-instance of \MAPFShort.

For the reverse direction, assume that $\langle G'',A,s'_0,t',2d\rangle$ is a yes-instance of \MAPFShort and let $s'_1, \dots , s'_{2d}$ be the witnessing feasible solution.
Every path agent $a_i$ performs a walk starting in $s_i$ and ending in $t_i$.
We can see this walk as a subdivision of a path in $G$ that we further transform into a path $P'_i$ in $G$ by removing any cycles and repeated vertices.
It is easily seen that the length of $P'_i$ is at most $d$.
Therefore, it suffices to show that $P'_1, \dots, P'_k$ are edge-disjoint.

\begin{claim}
For every edge $e \in E$ of $G$, there exists at most one $c_e \in [2d-1]$ such that the vertex $v_e = v_e^{2d-1}$ is not occupied by any agent of $B_e$ at time $c_e$.
\end{claim}
\begin{proofclaim}
For every agent $b_e^i$, the distance between $s'_0(b_e^i)$ and $t'(b_e^i)$ is exactly $2d-1$ and moreover, any walk other than the shortest path has length at least $2d + 1$.
Therefore, the agent $b_e^i$ must advance by one step towards its target in every turn up to one exception when he may remain stationary.
As a consequence, the edge agents $B_e$ do not leave the path $P_e$ and their order along $P_e$ remains unchanged.
In particular, they cannot swap.

Now, assume for a contradiction that there are two turns $c_e, c'_e \in [2d-1]$ in which the vertex $v_e^{2d-1}$ is not occupied by any agent of $B_e$.
Without loss of generality, let $c_e < c'_e$.
First, consider the agent $b_e^{2d-c_e-1}$.
If it moved without any delays, it would arrive at vertex $v_e^{2d-1}$ exactly at time $c_e$.
Thus, we have $s'_{c_e}(b_e^{2d-c_e-1}) = v_e^j$ where $j < 2d-1$.
On the other hand, if $j < 2d -2$ then the agent $b_e^{2d-c_e-1}$ cannot make it to its final destination $v_e^{4d-c_e-2}$ in the remaining $2d - c_e$ turns.
Therefore, it must be that $s'_{c_e}(b_e^{2d-c_e-1}) = v_e^{2d - 2}$.
In fact, we can show that $s'_{c_e}(b_e^{i}) = v_e^{i + c_e - 1}$ for every $i \in [2d-c_e-1]$.
It holds since the edge agents in $B_e$ must keep their relative order along $P_e$ and additionally, they cannot be more than $2d-c_e$ steps away from their targets.

However, this also means that all the agents $b_e^i$ for $i \in [2d-c_e-1]$ cannot make any further delays in order to reach their targets in time.
In particular, the agent $b_e^{2d - c'_e}$ occupies the vertex $v_e^{2d-1}$ at time $c'_e$ which is a contradiction.
\end{proofclaim}

It follows that at most one path agent can use the vertex $v_e$ throughout all $2d$ turns.
Therefore, the paths $P'_1, \dots, P'_k$ are edge-disjoint and $\langle G, s, t, k, d \rangle$ is a yes-instance of \textsc{BEUP}.

The described reduction together with~\cite{GolovachT11} takes care of the cases when $p$ is even.
For odd $p \ge 20$, it is sufficient to perform the same reduction from an instance $\langle G, s, t, k, \frac{p-1}{2} \rangle$ of \textsc{BEUP} and additionally, attach a leaf $v_a$ to every original starting position $s'_0(a)$ and set new a start $s'_0(a) = v_a$.
This modification does not increase the cliquewidth of the graph $G''$ and it is straightforward to prove its correctness. 
\end{proof}


\subsection{Maximum degree plus makespan}

\thmMAPFisNPCPlanarMakespanDelta*

\begin{figure}[bt]
    \centering

    \subfloat[Variable gadget $G^x$, $m=m(x)$.]{
    \scalebox{1.1}{
    \begin{tikzpicture}[inner sep=0.5mm]

    \node[draw,black,circle,line width=1pt,fill=white] (v16) at (2,0)[label=above right:{\small $v^x_{1,6}$}]{};
    \node[draw,black,circle,line width=1pt,fill=blue] (v11) at (3,0)[label=right:{\small$v^x_{1,1}$}]{};
    \node[draw,black,circle,line width=1pt,fill=white] (v12) at (3.5,0.87)[label=right:{\small$v^x_{1,2}$}]{};
    \node[draw,black,circle,line width=1pt,fill=red] (v13) at (3,1.73)[label=right:{\small$v^x_{1,3}$}]{};
    \node[draw,black,circle,line width=1pt,fill=white] (v14) at (2,1.73)[label=below right:{\small$v^x_{1,4}$}]{};
    \node[draw,black,circle,line width=1pt,fill=white] (v26) at (2,1.73)[label=above left:{\small$v^x_{2,6}$}]{};
    \node[draw,black,circle,line width=1pt,fill=white] (v15) at (1.5,0.87)[label=left:{\small$v^x_{1,5}$}]{};
    \node[draw,black,circle,line width=1pt,fill=blue] (v21) at (2.64,2.5)[label=right:{\small$v^x_{2,1}$}]{};
    \node[draw,black,circle,line width=1pt,fill=white] (v22) at (2.3,3.44)[label=right:{\small$v^x_{2,2}$}]{};
    \node[draw,black,circle,line width=1pt,fill=red] (v23) at (1.31,3.62)[label=above:{\small$v^x_{2,3}$}]{};
    \node[draw,black,circle,line width=1pt,fill=white] (v24) at (0.67,2.85)[label=left:{\small$v^x_{2,4}$}]{};
    \node[draw,black,circle,line width=1pt,fill=white] (v25) at (1.01,1.9)[label=left:{\small$v^x_{2,5}$}]{};
    \node[draw,black,circle,line width=1pt,fill=white] (vm5) at (1.02,-0.18)[label=left:{\small$v^x_{m,5}$}]{};
    \node[draw,black,circle,line width=1pt,fill=white] (vm6) at (0.68,-1.11)[label=left:{\small$v^x_{m,6}$}]{};
    \node[draw,black,circle,line width=1pt,fill=blue] (vm1) at (1.33,-1.87)[label=left:{\small$v^x_{m,1}$}]{};
    \node[draw,black,circle,line width=1pt,fill=white] (vm2) at (2.31,-1.7)[label=right:{\small$v^x_{m,2}$}]{};
    \node[draw,black,circle,line width=1pt,fill=red] (vm3) at (2.64,-0.76)[label=right:{\small$v^x_{m,3}$}]{};
    \node [label={[xshift=1.9cm, yshift=-0.55cm]{\small$v^x_{m,4}$}}] {};

    \node[draw,black,circle,line width=1pt,fill=black] () at (0,1.8)[]{};
    \node[draw,black,circle,line width=1pt,fill=black] () at (-0.25,1.75)[]{};
    \node[draw,black,circle,line width=1pt,fill=black] () at (-0.4,1.5)[]{};

    \node[draw,black,circle,line width=1pt,fill=black] () at (0,-0.08)[]{};
    \node[draw,black,circle,line width=1pt,fill=black] () at (-0.25,-0.04)[]{};
    \node[draw,black,circle,line width=1pt,fill=black] () at (-0.4,0.14)[]{};

    \draw[line width=1pt,black] (v16) -- (v11);
    \draw[line width=1pt,black] (v11) -- (v12);
    \draw[line width=1pt,black] (v12) -- (v13);
    \draw[line width=1pt,black] (v13) -- (v14);
    \draw[line width=1pt,black] (v14) -- (v15);
    \draw[line width=1pt,black] (v15) -- (v16);

    \draw[line width=1pt,black] (v14) -- (v21);
    \draw[line width=1pt,black] (v21) -- (v22);
    \draw[line width=1pt,black] (v22) -- (v23);
    \draw[line width=1pt,black] (v23) -- (v24);
    \draw[line width=1pt,black] (v24) -- (v25);
    \draw[line width=1pt,black] (v25) -- (v14);

    \draw[line width=1pt,black] (v16) -- (vm3);
    \draw[line width=1pt,black] (vm3) -- (vm2);
    \draw[line width=1pt,black] (vm2) -- (vm1);
    \draw[line width=1pt,black] (vm1) -- (vm6);
    \draw[line width=1pt,black] (vm6) -- (vm5);
    \draw[line width=1pt,black] (vm5) -- (v16);

    \end{tikzpicture}
    }
    }\hspace{20pt}
    \subfloat[Clause gadget $H^C$.]{
    \scalebox{1}{
    \begin{tikzpicture}[inner sep=0.3mm]
    \node[draw,black,circle,line width=1pt,fill=white] (v15) at (2,2)[label=left:$v^c_{1,5}$]{};
    \node[draw,black,circle,line width=1pt,fill=white] (v14) at (2,1)[label=left:$v^c_{1,4}$]{};
    \node[draw,black,circle,line width=1pt,fill=white] (v13) at (2,0)[label=left:$v^c_{1,3}$]{};
    \node[draw,black,circle,line width=1pt,fill=white] (v12) at (3,0)[label=above right:$v^c_{1,2}$]{};
    \node[draw,black,circle,line width=1pt,fill=white] (v11) at (3,1)[label=above:$v^c_{1,1}$]{};
    \node[draw,black,circle,line width=1pt,fill=white] (v22) at (3,-1)[label=below right:$v^c_{2,2}$]{};
    \node[draw,black,circle,line width=1pt,fill=white] (v23) at (2,-2)[label=below:$v^c_{2,3}$]{};
    \node[draw,black,circle,line width=1pt,fill=white] (v24) at (3,-2)[label=below:$v^c_{2,4}$]{};
    \node[draw,black,circle,line width=1pt,fill=white] (v25) at (4,-2)[label=below:$v^c_{2,5}$]{};
    \node[draw,black,circle,line width=1pt,fill=white] (v21) at (2,-1)[label=left:$v^c_{2,1}$]{};
    \node[draw,black,circle,line width=1pt,fill=white] (v32) at (4,0)[label=above right:$v^c_{3,2}$]{};
    \node[draw,black,circle,line width=1pt,fill=white] (v33) at (5,0)[label=right:$v^c_{3,3}$]{};
    \node[draw,black,circle,line width=1pt,fill=white] (v34) at (5,1)[label=right:$v^c_{3,4}$]{};
    \node[draw,black,circle,line width=1pt,fill=white] (v35) at (5,2)[label=right:$v^c_{3,5}$]{};
    \node[draw,black,circle,line width=1pt,fill=white] (v31) at (4.5,-0.8)[label=right:$v^c_{3,1}$]{};
    \node[draw,black,circle,line width=1pt,fill=white] (vt) at (3.5,-0.5)[label=below right:$v^c_{t}$]{};

    \draw[line width=1pt,black] (v15) -- (v14);
    \draw[line width=1pt,black] (v14) -- (v13);
    \draw[line width=1pt,black] (v13) -- (v12);
    \draw[line width=1pt,black] (v12) -- (v11);
    \draw[line width=1pt,black] (v12) -- (v22);
    \draw[line width=1pt,black] (v21) -- (v22);
    \draw[line width=1pt,black] (v22) -- (v23);
    \draw[line width=1pt,black] (v23) -- (v24);
    \draw[line width=1pt,black] (v24) -- (v25);
    \draw[line width=1pt,black] (v22) -- (vt);
    \draw[line width=1pt,black] (vt) -- (v32);
    \draw[line width=1pt,black] (v12) -- (v32);
    \draw[line width=1pt,black] (v32) -- (v31);
    \draw[line width=1pt,black] (v32) -- (v33);
    \draw[line width=1pt,black] (v33) -- (v34);
    \draw[line width=1pt,black] (v34) -- (v35);
    \draw[line width=1pt,black] (v11) -- (v12);
    \end{tikzpicture}
}
}

\caption{The two gadgets used in the proof of Theorem~\ref{thm:MAPFisNPC:PlanarMakespanDelta}. In subfigure (a), the color red (blue resp.) is used on the vertices that may be incident to clause gadgets, in which $x$ appear as a negative (positive resp.) literal.}
\label{figure:planar-hard}
\end{figure}
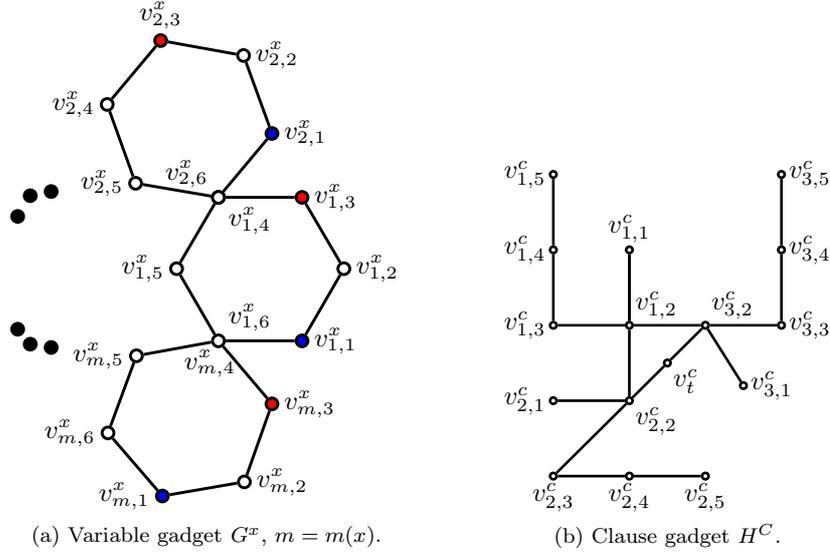

\begin{proof}
    The problem is clearly in NP. To show that is is also NP-hard, we present a reduction from the \textsc{Planar $3$-SAT} problem which is known to be NP-complete~\cite{GJ79}.
    In that problem, a 3CNF formula $\phi$ on $n$ variables and $m$ clauses is given as an input. Let $X$ denote the set of variables appearing in $\phi$. We say that a bipartite graph $G'=(V,U,E)$ \emph{corresponds} to $\phi$ if it is constructed from $\phi$ in the following way: for each variable $x_i \in X$ add the \emph{variable vertex} $v_i$ in $V$ and for each clause $C_j$ of $\phi$ add a \emph{clause vertex} $c_j$ in $U$. Then the edge $v_ic_j$ is added if the variable $x_i$ appears in a literal of the clause $C_j$.
    A 3CNF formula $\phi$ is valid as input to the \textsc{Planar 3-SAT} problem if the graph $G'$ that corresponds to $\phi$ is planar.
    The question is whether there exists a truth assignment to the variables of $X$ satisfying $\phi$.

    Our construction starts with a planar embedding of the bipartite graph $G'$ that corresponds to $\phi$.
    For each variable $x\in X$, let $m(x)$ be the number clauses that contain $x$ or $\lnot x$.
    First, we replace each clause vertex $c$ by $H^c$, which is a copy of the \emph{clause gadget} $H$ (illustrated in Figure~\ref{figure:planar-hard}(a)). Then, we replace each variable vertex $x$ by $G^x$, a copy of the the \emph{variable gadget} $G^x$ (illustrated in Figure~\ref{figure:planar-hard}(b)). The superscripts $x$ and $c$ will be used to differentiate between the vertices of the variable and clause gadgets respectively. Note that all the edges of $G'$ have been removed at this stage.
    So, for each edge $xc$ in $G'$, we add two new edges connecting vertices of $G^x$ and $H^c$.
    In particular, for each clause gadget $H^c$, we fix an arbitrary ordering of the literals appearing in $C$. Then, for each $xc\in E(G')$, we assign (arbitrarily) an $i \in [m(x)]$ that hasn't already been assigned to a different literal appearing in $C$. Let $x$ be in the $j$-th literal of $C$.
    If this literal is a positive one, then we add the edges $v^x_{i,1}v^c_{j,1}$ and $v^x_{i,1}v^c_{j,5}$.
    Otherwise, we add the edges $v^x_{i,3}v^c_{j,1}$ and $v^x_{i,3}v^c_{j,5}$. Let $G$ be the resulting graph. Note here that the arbitrary choices made for these edges will only impact the planarity of $G$.
    Thus, we will specify them in the final stage of the proof.

    We continue by defining the set of agents $A$ and the functions $s_0$ and $t$.
    First, for each clause $C$, we create four \emph{clause agents}, three of them denoted as $a^c_i$, $i \in [3]$, and the final as $a^c$.
    For $i \in [3]$, we set $s_0(a^c_i) = v^c_{i,1}$ and $t(a^c_i) = v^c_{i,4}$.
    Also, $s_0(a^c) = v^c_{1,2}$ and $t(a^c) = v^c_t$.
    Next, for each variable $x$, we create $m(x)$ \emph{variable agents} $a^x_i$, $i \in [m(x)]$.
    For each $i \in [m(x)]$, we set $s_0(a^x_i) = v^x_{i,2}$ and $t(a^x_i) = v^x_{i,5}$.
    This completes our construction. Observe that, for each $C$, if the agents $a^c_i$, $i\in[3]$, move only through edges of $H^c$, then any feasible solution will have a makespan of at least $4$.

    We are now ready to show that $\phi$ is a yes-instance of \textsc{Planar $3$-SAT} if and only if $\ell(\langle G,A,s_0,t\rangle)\leq 3$. First, assume that we have a satisfying assignment $\sigma$ for $\phi$. We will show that there exists a feasible solution $s_1,s_2,s_3$ for $\langle G,A,s_0,t\rangle$.
    For each variable $x$ that appears in $m$ clauses, if $\sigma(x) = true$ then, for all agents $a^x_i$, $i \in [m]$, we set $s_j(a^x_i) = v^x_{i,2+j}$ for all $j\in [3]$.
    Otherwise, ($\sigma(x) = false$), for all $i \in [m]$, we set $s_1(a^x_i) = v^x_{i,1}$, $s_2(a^x_i) = v^x_{i,6}$ and $s_3(a^x_i) = v^x_{i,5}$ (note that $v^x_{1,6}=v^x_{m,4}$). 
    Next, for each clause $C$, we deal with the agents $a^c_i$, $i\in [3]$, and $a^c$.
    Since $\sigma$ is a satisfying assignment, there exists at least one literal $x$ in $C$ that satisfies it. W.l.o.g., assume that $x$ is the $j$-th literal of $C$. For the agents $a^c_i$, $i \in [3] \setminus \{j\}$
    we set $s_1(a^c_i) = v^c_{i,2}$, $s_2(a^c_i) = v^c_{i,3}$ and $s_3(a^c_i) = v^c_{i,4}$.
    For the agent $a^c_j$, we set $s_1(a^c_i) = u$, $s_2(a^c_i) = v^c_{j,5}$ and
    $s_3(a^c_i) = v^c_{j,4}$, where $u$ is the common neighbor of $v^c_{j,1}$ and $v^c_{j,5}$ (which belongs in some variable gadget).
    Finally, we define the movement of the agent $a^c$. If $j=1$ then we set
    $s_1(a^c) = v^c_{1,2} $, $s_2(a^c) = v^c_{2,2}$ and $s_3(a^c) = v^c_t$.
    Otherwise, $s_1(a^c) = v^c_{j,2} $, $s_2(a^c) = v^c_t$ and $s_3(a^c) = v^c_t$.

    Observe that for any agent $a \in A$, $s_{i}(a)$ is either a neighbor of $s_{i-1}(a)$, or the same as $s_{i-1}(a)$. Therefore, it suffices to show that no two agents are occupying the same vertex during the same turn.
    By the definition of the sequence $s_1,s_2,s_3$, this may only happen between a clause and a variable agent and only during the first turn.
    Let $a$ be a clause agent that starts at a clause gadget $H^c$, and $s_1(a)=u\in G^x$.
    Notice that, by the definition of $s$, the variable $x$ has received by $\sigma$ a truth value such that the literal containing $x$ satisfies the clause $C$.
    We consider two cases, either this literal is $x$ or $\neg x$.
    In the first case, $\sigma(x)=true$; therefore, for all $i \in [m(x)]$,
    $s_1 (a^x_i) = v^x_{i,3}$. Also, $s_1 (a) = v^x_{i,1}$ since the represented literal
    is a positive one and by the construction of $G$. In the latter case, $\sigma(x)=false$; therefore, for all $i \in [m(x)]$, $s_1 (a^x_i) = v^x_{i,1}$.
    Also, $s_1 (a) = v^x_{i,3}$, since the represented literal is a negative one and by the construction of $G$. In both cases, there is no collision. This finishes the first direction of the reduction.

    For the reverse direction, we show that, if $\ell(\langle G,A,s_0,t\rangle)\leq 3$, then there exists a satisfying assignment of $\phi$.
    We start with some observations.
    \begin{observation}
        In any solution of makespan $3$, any agent $a \in A \setminus \{a^c \mid c \in \phi \}$
        must move through a shortest path between $s_0 (a)$ and $t (a)$.
    \end{observation}
    Indeed, for any agent in the given set, the distance between its starting and terminal positions is exactly $3$.
    \begin{observation}
        Let $A_x$ be the set of agents $\{a^x_i \mid i \in [m(x)]\}$  for a variable $x$.
        In any solution of makespan $3$ either
        $(s_1(a^x_i),s_2(a^x_i), s_3(a^x_i)) = (v_{i,3}^x,v_{i,4}^x,v_{i,5}^x)$ for all $i \in [m(x)]$
        or $(s_1(a^x_i),s_2(a^x_i), s_3(a^x_i)) = (v_{i,1}^x,v_{i,6}^x,v_{i,5}^x)$ for all  $i \in [m(x)]$ (recall that $v_{i,6}^x=v_{m,4}$ for $m=m(x)$).
    \end{observation}
    Indeed, if this were not true, then there would exist an $i\in [m(x)]\setminus \{1\}$ such that $s_2(a^x_i)=s_2(a^x_{i-1})$, or $s_2(a^x_1)=s_2(a^x_m)$, for $m=m(x)$.
    \begin{observation}
        For any clause $C$, at most two of the agents $a^c_i$, $i \in [3]$, can have $s(a^c_i) = v^c_{i,2}$.
    \end{observation}
    This follows directly from the fact that if the agents of $H^c$ move only through edges of $H^c$, then any feasible solution will have a makespan of at least $4$. In other word, at least one clause agent of each clause must move through the vertices of a variable gadgets.

    We define the following assignment: let $x$ be a variable; we set $x$  to be true if $s_1 (a^x_1) = v^x_{1,3}$ and false otherwise. This assignment satisfies $\phi$.
    Indeed, consider a clause $C$. We know that there exists a $j \in [3]$ such that $a^c_j$ moves through a vertex $u$ belonging in $G^x$. There are two cases to be analysed.

    \textbf{Case 1 ($\boldsymbol{u = v^x_{i,1}}$)}: In this case, by construction, we know that $x$ appears positively in the clause $C$.
    From the previous observations, it follows that $s_1 (a^c_j) = u = v^x_{i,1}$ and thus
    $s_1 (a^x_i) = u = v^x_{i,3}$. Therefore, $s_1 (a^x_1) = u = v^x_{1,3}$ and $x$ has been set to true.

    \textbf{Case 2 ($\boldsymbol{u = v^x_{i,3}}$)}: In this case, by construction, we know that $x$ appears negatively in the clause $C$.
    From the previous observations, it follows that $s_1 (a^c_j) = u = v^x_{i,3}$ and thus
    $s_1 (a^x_i) = u = v^x_{i,1}$. Therefore, $s_1 (a^x_1) = u = v^x_{1,1}$ and $x$ has been set to false.

    In both cases above, the clause $C$ is satisfied. In other words, the assignment we defined is indeed satisfying $\phi$.

    Lastly, it remains to argue about the planarity of $G$.

    \begin{claim}
     The graph $G$ can be constructed so that it is planar.
    \end{claim}
    \begin{proofclaim}
    In order to draw a planar embedding of $G$, we start from a planar emending of $G'$. Imagine that around each vertex of $G'$ we have a circle that is small enough so that no two such circles intersect. We may also assume that the edges of $G'$ are drawn in such a way so that each edge intersects every such circle at most once.
    We will say that the area inside any such circle is \emph{close} to the corresponding vertex and the area not included in any cycle is \emph{far} from the vertices.

    First, when we replace the vertices with gadgets, we draw them inside the corresponding circles, and as they are drawn in Figure~\ref{figure:planar-hard}.
    We also need to draw the new edges. Recall that each edge of $G'$ has been replaced by two new edges in $G$.
    We draw the new edges so that their parts that lie far from the vertices are parallel to the original edges. That way, and since $G'$ is planar, no two of the new edges are crossing in the area far from the vertices.

    Before we deal with the part of the new edges that lies in the area close to the vertices (defined by the aforementioned circles) we will define,
    for each gadget (regardless of it being a variable or a clause gadget) an ordering of the edges that need to be attached to
    its vertices. We proceed as follows. For each vertex of $G$, consider the center of its circle.
    We define a ray $\mu$ that starts from this center and moves horizontally to the right of it. We enumerate the edges $e_1,\dots,e_m$ of $G'$ based on their entering points on this circle (i.e., the intersection of the original edges and the circles around the original vertices in $G'$).
    In particular, we follow the order in which $\mu$ meets the entering points of the edges when it moves counter clockwise.
    Let us denote the edges that replace the edge $e_i$ (for any $i\in[m]$) by $e_{i,1}$ and $e_{i,2}$.
    We decide the vertices of the gadget to which these edges will be incident based on $i$.
    When we consider a variable gadget, $e_{i,1}$ and $e_{i,2}$ will be incident to $v^x_{i,1}$ if the corresponding literal is a positive one and to $v^x_{i,3}$ otherwise (as described in the initial construction of $G$). Similarly, when we consider a clause gadget,  $e_{i,1}$ and $e_{i,2}$ will be incident to $v^c_{i,1}$ and $v^c_{i,5}$ respectively.

    To show that we can draw the new edges so they do not cross in the area close to the circles, it suffices to observe that:
    \begin{observation}
        Consider two concentric circles. Let $p_1,\ldots, p_i$ be $i\ge 1$ vertices on the outside circle in a counter-clockwise order. Also, let $p'_1,\ldots, p'_i$ be $i$ vertices on the inside circle in a counter-clockwise order. Then, we can draw all the edges  $p_ip'_i$ such that all of then are in the area between the two circles and no two of them are crossing.
    \end{observation}
    This completes the proof of the claim.
    \end{proofclaim}
\end{proof}
\thmMAPFisNPCPlanarMakespanTwo*
\begin{figure}[!t]
\centering

\subfloat[Variable gadget $G^x$, $m=m(x)$.]{
\scalebox{0.7}{
\begin{tikzpicture}[inner sep=0.7mm]

\node [draw, circle, line width=1pt, fill=white] (u1) at  (5,-1)  [label=above:\LARGE $u^x_{1,1}$]{};
\node [draw, circle, line width=1pt, fill=white] (w1) at  (3,0)  [label=above:\LARGE $w^x_{1}$]{};
\node [draw, circle, line width=1pt, fill=white] (u3) at  (1,-1)  [label=above:\LARGE $u^x_{1,3}$]{};
\node [draw, circle, line width=1pt, fill=blue] (u2) at  (3,-2)  [label=above:\LARGE $u^x_{1,2}$]{};
\node [draw, circle, line width=1pt, fill=white] (u4) at  (0,-2)  [label=above:\LARGE $u^x_{1,4}$]{};
\node [draw, circle, line width=1pt, fill=white] (u5) at  (1,-3)  [label=above:\LARGE $u^x_{1,5}$]{};
\node [draw, circle, line width=1pt, fill=white] (v1) at  (5,1)  [label=above:\LARGE $v^x_{1,1}$]{};
\node [draw, circle, line width=1pt, fill=white] (v3) at  (1,1)  [label=above:\LARGE $v^x_{1,3}$]{};
\node [draw, circle, line width=1pt, fill=red] (v2) at  (3,2)  [label=above:\LARGE $v^x_{1,2}$]{};
\node [draw, circle, line width=1pt, fill=white] (v4) at  (0,2)  [label=above:\LARGE $v^x_{1,4}$]{};
\node [draw, circle, line width=1pt, fill=white] (v5) at  (1,3)  [label=right:\LARGE $v^x_{1,5}$]{};

\draw[-, line width=1pt]  (v1) -- (v2);
\draw[-, line width=1pt]  (v2) -- (v3);
\draw[-, line width=1pt]  (v3) -- (v4);
\draw[-, line width=1pt]  (v4) -- (v5);
\draw[-, line width=1pt]  (v5) -- (v2);
\draw[-, line width=1pt]  (w1) -- (v1);
\draw[-, line width=1pt]  (w1) -- (v3);

\draw[-, line width=1pt]  (u1) -- (u2);
\draw[-, line width=1pt]  (u2) -- (u3);
\draw[-, line width=1pt]  (u3) -- (u4);
\draw[-, line width=1pt]  (u4) -- (u5);
\draw[-, line width=1pt]  (u5) -- (u2);
\draw[-, line width=1pt]  (w1) -- (u1);
\draw[-, line width=1pt]  (w1) -- (u3);

\node [draw, circle, line width=1pt, fill=white] (u1) at  (0.2675,8.609)  [label=above left:\LARGE $v^x_{1,1}$]{};
\node [draw, circle, line width=1pt, fill=white] (w1) at  (-0.159,6.414)  [label=above left:\LARGE $w^x_{1}$]{};
\node [draw, circle, line width=1pt, fill=white] (u3) at  (-2.171,5.438)  [label=above left:\LARGE $v^x_{1,3}$]{};
\node [draw, circle, line width=1pt, fill=red] (u2) at  (-1.744,7.633)  [label=above left:\LARGE $v^x_{1,2}$]{};
\node [draw, circle, line width=1pt, fill=white] (u4) at  (-3.573,5.255)  [label=above left:\LARGE $v^x_{1,4}$]{};
\node [draw, circle, line width=1pt, fill=white] (u5) at  (-3.756,6.658)  [label=above left:\LARGE $v^x_{1,5}$]{};
\node [draw, circle, line width=1pt, fill=white] (v1) at  (1.853,7.39)  [label=above left:\LARGE  $u^x_{1,1}$]{};
\node [draw, circle, line width=1pt, fill=white] (v3) at  (-0.585,4.219)  [label=above left:\LARGE   $u^x_{1,3}$]{};
\node [draw, circle, line width=1pt, fill=blue] (v2) at  (1.426,5.195)  [label=above left:\LARGE   $u^x_{1,2}$]{};
\node [draw, circle, line width=1pt, fill=white] (v4) at  (-0.402,2.817)  [label=above left:\LARGE  $u^x_{1,4}$]{};
\node [draw, circle, line width=1pt, fill=white] (v5) at  (1,3)  [label=above left:{\LARGE  $u^x_{1,5}$}]{};

\draw[-, line width=1pt]  (v1) -- (v2);
\draw[-, line width=1pt]  (v2) -- (v3);
\draw[-, line width=1pt]  (v3) -- (v4);
\draw[-, line width=1pt]  (v4) -- (v5);
\draw[-, line width=1pt]  (v5) -- (v2);
\draw[-, line width=1pt]  (w1) -- (v1);
\draw[-, line width=1pt]  (w1) -- (v3);

\draw[-, line width=1pt]  (u1) -- (u2);
\draw[-, line width=1pt]  (u2) -- (u3);
\draw[-, line width=1pt]  (u3) -- (u4);
\draw[-, line width=1pt]  (u4) -- (u5);
\draw[-, line width=1pt]  (u5) -- (u2);
\draw[-, line width=1pt]  (w1) -- (u1);
\draw[-, line width=1pt]  (w1) -- (u3);

\node [draw, circle, line width=1pt, fill=white] (u1) at  (1.853,-7.39)  [label=above right: \LARGE  $v^x_{m,1}$]{};
\node [draw, circle, line width=1pt, fill=white] (w1) at  (-0.159,-6.414)  [label=above right:\LARGE  $w^x_{m}$]{};
\node [draw, circle, line width=1pt, fill=white] (u3) at  (-0.585,-4.219)  [label=above right:\LARGE  $v^x_{m,3}$]{};
\node [draw, circle, line width=1pt, fill=red] (u2) at  (1.426,-5.195)  [label=above right:\LARGE  $v^x_{m,2}$]{};
\node [draw, circle, line width=1pt, fill=white] (u4) at  (-0.402,-2.817)  [label=left:\LARGE  $v^x_{m,4}$]{};
\node [draw, circle, line width=1pt, fill=white] (u5) at  (1,-3)  [label= right:\LARGE  $v^x_{m,5}$]{};
\node [draw, circle, line width=1pt, fill=white] (v1) at  (0.267,-8.609)  [label=above right:\LARGE  $u^x_{m,1}$]{};
\node [draw, circle, line width=1pt, fill=white] (v3) at  (-2.171,-5.438)  [label=above right:\LARGE  $u^x_{m,3}$]{};
\node [draw, circle, line width=1pt, fill=blue] (v2) at  (-1.744,-7.633)  [label=above right:\LARGE  $u^x_{m,2}$]{};
\node [draw, circle, line width=1pt, fill=white] (v4) at  (-3.573,-5.255)  [label=above right:\LARGE  $u^x_{m,4}$]{};
\node [draw, circle, line width=1pt, fill=white] (v5) at  (-3.756,-6.658)  [label=above right:\LARGE  $u^x_{m,5}$]{};

\draw[-, line width=1pt]  (v1) -- (v2);
\draw[-, line width=1pt]  (v2) -- (v3);
\draw[-, line width=1pt]  (v3) -- (v4);
\draw[-, line width=1pt]  (v4) -- (v5);
\draw[-, line width=1pt]  (v5) -- (v2);
\draw[-, line width=1pt]  (w1) -- (v1);
\draw[-, line width=1pt]  (w1) -- (v3);

\draw[-, line width=1pt]  (u1) -- (u2);
\draw[-, line width=1pt]  (u2) -- (u3);
\draw[-, line width=1pt]  (u3) -- (u4);
\draw[-, line width=1pt]  (u4) -- (u5);
\draw[-, line width=1pt]  (u5) -- (u2);
\draw[-, line width=1pt]  (w1) -- (u1);
\draw[-, line width=1pt]  (w1) -- (u3);

    \node[draw,black,circle,line width=0.5pt,fill=black] () at (-4.25,4.25)[]{};
    \node[draw,black,circle,line width=0.5pt,fill=black] () at (-5.25,3.675)[]{};
    \node[draw,black,circle,line width=0.5pt,fill=black] () at (-5.75,2.75)[]{};

    \node[draw,black,circle,line width=0.5pt,fill=black] () at (-4.25,-4.25)[]{};
    \node[draw,black,circle,line width=0.5pt,fill=black] () at (-5.25,-3.675)[]{};
    \node[draw,black,circle,line width=0.5pt,fill=black] () at (-5.75,-2.75)[]{};

\end{tikzpicture}
}
}\hspace{10pt}
\subfloat[Clause gadget $H^C$.]{
    \scalebox{1}{
    \begin{tikzpicture}[inner sep=0.5mm]
    \node[draw,black,circle,line width=1pt,fill=white] (v13) at (2,0)[label=above:$v^c_{1,3}$]{};
    \node[draw,black,circle,line width=1pt,fill=white] (v12) at (3,0)[label=above right:$v^c_{1,2}$]{};
    \node[draw,black,circle,line width=1pt,fill=white] (v11) at (3,1)[label=above:$v^c_{1,1}$]{};
    \node[draw,black,circle,line width=1pt,fill=white] (v22) at (3,-1)[label=below right:$v^c_{2,2}$]{};
    \node[draw,black,circle,line width=1pt,fill=white] (v23) at (2,-2)[label=below:$v^c_{2,3}$]{};
    \node[draw,black,circle,line width=1pt,fill=white] (v21) at (2,-1)[label=left:$v^c_{2,1}$]{};
    \node[draw,black,circle,line width=1pt,fill=white] (v32) at (4,0)[label=below right:$v^c_{3,2}$]{};
    \node[draw,black,circle,line width=1pt,fill=white] (v33) at (5,0)[label=above:$v^c_{3,3}$]{};
    \node[draw,black,circle,line width=1pt,fill=white] (v31) at (4,1)[label=above:$v^c_{3,1}$]{};

    \draw[line width=1pt,black] (v13) -- (v12);
    \draw[line width=1pt,black] (v12) -- (v11);
    \draw[line width=1pt,black] (v12) -- (v22);
    \draw[line width=1pt,black] (v21) -- (v22);
    \draw[line width=1pt,black] (v22) -- (v23);
    \draw[line width=1pt,black] (v22) -- (v32);
    \draw[line width=1pt,black] (v12) -- (v32);
    \draw[line width=1pt,black] (v32) -- (v31);
    \draw[line width=1pt,black] (v32) -- (v33);
    \draw[line width=1pt,black] (v11) -- (v12);
    \end{tikzpicture}
}
}

\caption{The two gadgets used in the proof of Theorem~\ref{thm:NP-hard-makespan-2}. In subfigure (a), the color red (blue resp.) is used on the vertices that may be incident to clause gadgets, in which $x$ appear as a negative (positive resp.) literal.}
\label{figure:planar-hard-noswaps-makespan2}

\end{figure}
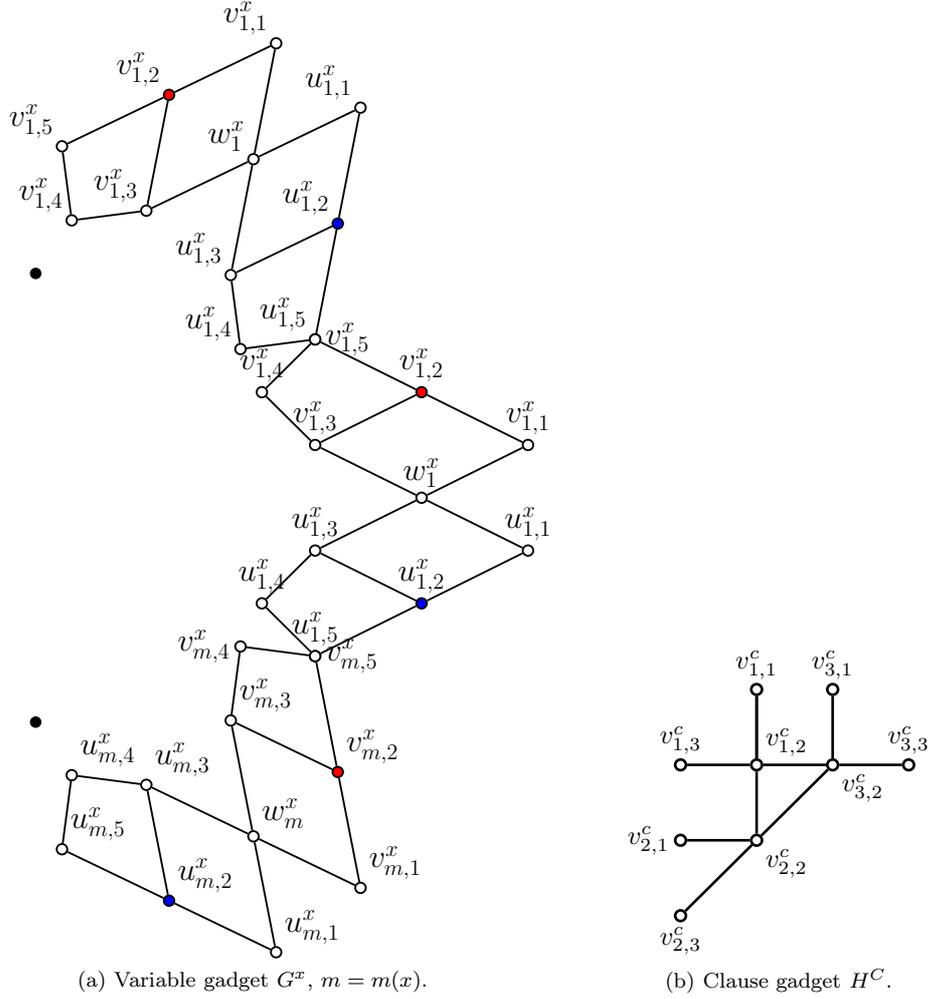

\begin{proof}
    The proof is similar to that of Theorem~\ref{thm:MAPFisNPC:PlanarMakespanDelta}. The problem is clearly in \NP.
    To show that is is also NP-hard, we present a reduction from the \textsc{Planar $3$-SAT}.
    As before, we assume that we have the 3CNF formula $\phi$ on $n$ variables and $m$ clauses. We denote with $X$ the set of variables appearing in $\phi$. Also, let $G'=(V,U,E)$ be the bipartite graph that \emph{corresponds} to $\phi$.

    Our construction starts with a planar embedding of the bipartite graph $G'$ that corresponds to $\phi$.
    For each variable $x\in X$, let $m(x)$ be the number of clauses that contain $x$ or $\lnot x$.
    First, we replace each clause vertex $c$ by $H^c$, which is a copy of the \emph{clause gadget} $H$ (illustrated in Figure~\ref{figure:planar-hard-noswaps-makespan2}(a)). Then, we replace each variable vertex $x$ by $G^x$, a copy of the the \emph{variable gadget} $G^x$ (illustrated in Figure~\ref{figure:planar-hard-noswaps-makespan2}(b)).
    The superscripts $x$ and $c$ will be used to differentiate between the vertices of the variable and clause gadgets respectively. Note that all the edges of $G'$ have been removed at this stage.
    So, for each edge $xc$ in $G'$, we add two new edges connecting vertices of $G^x$ and $H^c$.
    In particular, to each $xc\in E(G')$, we assign an $i \in [m(x)]$ that hasn't already been assigned to a different appearance of $x$ in $\phi$. Let $x$ be in the $j$-th literal of $C$.
    If this literal is a positive one, then we add the edges $u^x_{i,2}v^c_{j,1}$ and $u^x_{i,2}v^c_{j,3}$.
    Otherwise, we add the edges $v^x_{i,2}v^c_{j,1}$ and $v^x_{i,2}v^c_{j,3}$. Let $G$ be the resulting graph.

    We continue by defining the set of agents $A$ and the functions $s_0$ and $t$.
    First, for each clause $C$, we create three \emph{clause agents} denoted as $a^c_i$, $i \in [3]$.
    For $i \in [3]$, we set $s_0(a^c_i) = v^c_{i,1}$ and $t(a^c_i) = v^c_{i,3}$.
    Then, for each clause $C$, we create four \emph{clause agents}, three of them denoted as $a^c_i$, $i \in [3]$, and the final as $a^c$.
    For $i \in [3]$, we set $s_0(a^c_i) = v^c_{i,1}$ and $t(a^c_i) = v^c_{i,3}$.
    Also, $s_0(a^c) = v^c_{1,2}$ and $t(a^c) = v^c_2$.
    Next, for each variable $x$, we create $4m(x)$ \emph{variable agents} $a^x_{i,j}$ and $b^x_{i,j}$, $i \in [m(x)]$ and $j \in [2]$.
    For each $i \in [m(x)]$, we set $s_0(a^x_{i,1}) = v^x_{i,1}$, $s_0(b^x_{i,1}) = u^x_{i,1}$, $s_0(a^x_{i,2}) = v^x_{i,4}$
    and $s_0(b^x_{i,2}) = u^x_{i,4}$. Finally, for each $i \in [m(x)]$, we set $t(a^x_{i,1}) = v^x_{i,3}$, $t(b^x_{i,1}) = u^x_{i,3}$, $t(a^x_{i,2}) = v^x_{i,2}$ and $t(b^x_{i,2}) = u^x_{i,2}$.
    This completes our construction. Observe that, for each $C$, since we are not allowing swaps, if the agents $a^c_i$, $i\in[3]$, move only through edges of $H^c$, then any feasible solution will have a makespan of at least $3$.

    We are now ready to show that $\phi$ is a yes-instance of \textsc{Planar $3$-SAT} if and only if $\langle G,A,s_0,t,2 \rangle$ is a yes-instance of \MAPFNoSwap. First, assume that we have a satisfying assignment $\sigma$ for $\phi$. We will show that there exists a feasible solution $s_1,s_2$ for $\langle G,A,s_0,t,2\rangle$.
    For each variable $x$ that appears in $m$ clauses, if $\sigma(x) = true$ then, for all agents $a^x_{i,j}$ and $b^x_{i,j}$, $i \in [m]$ and $j \in [2]$, we set:
    \begin{itemize}
        \item $s_1(a^x_{i,1}) = v^x_{i,2}$ and $s_2(a^x_{i,1}) = v^x_{i,3}$
        \item $s_1(b^x_{i,1}) = w^x_{i}$ and $s_2(b^x_{i,1}) = u^x_{i,3}$
        \item $s_1(a^x_{i,2}) = v^x_{i,5}$ and $s_2(a^x_{i,2}) = v^x_{i,2}$
        \item $s_1(b^x_{i,2}) = v^x_{i,3}$ and $s_2(b^x_{i,1}) = u^x_{i,2}$
    \end{itemize}
    Otherwise, ($\sigma(x) = false$), for all $i \in [m]$, we set:
    \begin{itemize}
        \item $s_1(a^x_{i,1}) = w^x_{i}$ and $s_2(a^x_{i,1}) = v^x_{i,3}$
        \item $s_1(b^x_{i,1}) = u^x_{i,2}$ and $s_2(b^x_{i,1}) = u^x_{i,3}$
        \item $s_1(a^x_{i,2}) = v^x_{i,3}$ and $s_2(a^x_{i,2}) = v^x_{i,2}$
        \item $s_1(b^x_{i,2}) = v^x_{i,5}$ and $s_2(b^x_{i,1}) = u^x_{i,2}$
    \end{itemize}
    Note that, for any $i \in [m]$, $v^x_{i,5}=u^x_{k,5}$ where $k= 1+ (i \mod m)$. Therefore, no two agents occupy the same vertex during the same turn.
    Next, for each clause $C$, we deal with the agents $a^c_i$, $i\in [3]$, and $a^c$.
    Since $\sigma$ is a satisfying assignment, there exists at least one literal $\lambda$ in $C$ that satisfies it.
    W.l.o.g., assume that $x$ is the variable that appears in $\lambda$ and that $\lambda$ in  $j$-th literal of $C$.
    For the agents $a^c_i$, $i \in [3] \setminus \{j\}$
    we set $s_1(a^c_i) = v^c_{i,2}$ and $s_2(a^c_i) = v^c_{i,3}$.
    For the agent $a^c_j$, we set $s_1(a^c_i) = u$ and $s_2(a^c_i) = v^c_{j,5}$, where $u$ is the common neighbor of $v^c_{j,1}$ and $v^c_{j,3}$ (which belongs in the variable gadget of $x$).
    Finally, we define the movement of the agent $a^c$. We set $s_1(a^c) = v^c_{j,2} $ and $s_2(a^c) = v^c_{2,2}$.

    Observe that for any agent $a \in A$, $s_{i}(a)$ is either a neighbor of $s_{i-1}(a)$, or the same as $s_{i-1}(a)$. Therefore, it suffices to show that no two agents are occupying the same vertex during the same turn.
    By the definition of the sequence $s_1,s_2$, this may only happen between a clause and a variable agent and only during the first turn. Let $a$ be a clause agent that starts at a clause gadget $H^c$, and $s_1(a)=u\in G^x$.
    Notice that, by the construction of the solution, $\sigma(x)$ is such that the literal $\lambda$ containing $x$ satisfies the clause $C$.
    We consider two cases, either this literal is $\lambda=x$ or $\lambda=\neg x$.
    In the first case, $\sigma(x)=true$; therefore, for all $i \in [m(x)]$, no agent occupies the vertex $u^x_{i,2}$.
    Also, $s_1 (a) = u^x_{i,2}$ since the represented literal is a positive one and by the construction of $G$.
    In the latter case, $\sigma(x)=false$; therefore, for all $i \in [m(x)]$, no agent occupies the vertex $v^x_{i,2}$.
    Also, $s_1 (a) = v^x_{i,2}$, since the represented literal is a negative one and by the construction of $G$.
    In both cases, there is no collision. This finishes the first direction of the reduction.

    For the reverse direction, we show that, if $s_1,s_2$ is a solution of $\langle G,A,s_0,t,2 \rangle$, then there exists a satisfying assignment of $\phi$.
    We start with some observations.
    \begin{observation}
        In any solution of makespan $2$, any agent $a \in A \setminus \{a^c \mid c \in \phi \}$
        must move through a shortest path between $s_0 (a)$ and $t (a)$.
    \end{observation}
    Indeed, for any agent in the given set, the distance between its starting and terminal positions is exactly $2$.
    \begin{claim}\label{cl:makespan-2-variable-agents}
        Let $A_x$ be the set of agents $\{a^x_i, b^x_i \mid i \in [m(x)]\}$  for a variable $x$.
        In any solution of makespan $2$ either, for all $i \in [m(x)]$,:\\
    \begin{itemize}
        \item $s_1(a^x_{i,1}) = v^x_{i,2}$ and $s_2(a^x_{i,1}) = v^x_{i,3}$
        \item $s_1(b^x_{i,1}) = w^x_{i}$ and $s_2(b^x_{i,1}) = u^x_{i,3}$
        \item $s_1(a^x_{i,2}) = v^x_{i,5}$ and $s_2(a^x_{i,2}) = v^x_{i,2}$
        \item $s_1(b^x_{i,2}) = v^x_{i,3}$ and $s_2(b^x_{i,1}) = u^x_{i,2}$
    \end{itemize}
    or, for all $i \in [m(x)]$,:
    \begin{itemize}
        \item $s_1(a^x_{i,1}) = w^x_{i}$ and $s_2(a^x_{i,1}) = v^x_{i,3}$
        \item $s_1(b^x_{i,1}) = u^x_{i,2}$ and $s_2(b^x_{i,1}) = u^x_{i,3}$
        \item $s_1(a^x_{i,2}) = v^x_{i,3}$ and $s_2(a^x_{i,2}) = v^x_{i,2}$
        \item $s_1(b^x_{i,2}) = u^x_{i,5}$ and $s_2(b^x_{i,1}) = u^x_{i,2}$
    \end{itemize}
    \end{claim}
    \begin{proofclaim}
    Consider the movement of any agent $a^x_{i,1}$ and observe how this affects the rest of the agents.
    W.l.o.g. assume that, for some $i \in [m(x)]$, $a^x_{i,1}$ moves through $v^x_{i,2}$. We will show that:
    \begin{itemize}
        \item $s_1(a^x_{i,2}) = v^x_{i,5}$,
        \item $s_1(b^x_{k,1}) = u^x_{k,3}$,
        \item $s_1(b^x_{k,1}) = w^x_{k}$ and
        \item $s_1(a^x_{k,1}) = v^x_{k,2}$,
    \end{itemize}
    where $k = 1+ (i \mod m)$.

    Notice that $s_1(a^x_{i,1}) = v^x_{i,2}$ forces that $s_1(a^x_{i,2}) = v^x_{i,5}$, as otherwise we would not achieve a solution of makespan $2$. Similarly, since $v^x_{i,5} = u^x_{k,5}$, we have that $s_1(b^x_{k,1}) = u^x_{k,3}$. This gives us that $s_1(b^x_{k,1}) = w^x_{k}$ and, finally, $s_1(a^x_{k,1}) = v^x_{k,2}$.
    Then, the $s_2$ follows.

    Similarly, we can show that $s_1(a^x_{i,1}) = w^x_{i}$ forces the second set of moves.
    \end{proofclaim}

    \begin{observation}
        For any clause $C$, at most two of the agents $a^c_i$, $i \in [3]$, can have $s(a^c_i) = v^c_{i,2}$.
    \end{observation}
    This follows directly from the fact that if the agents of $H^c$ move only through edges of $H^c$, then any feasible solution will have a makespan of at least $3$. In other words, at least one clause agent of each clause must move through the vertices of a variable gadgets.

    We define the following truth assignment. Let $x$ be a variable; we set $x$ to be true if $s_1 (a^x_{1,1}) = v^x_{1,2}$ and false otherwise. This assignment satisfies $\phi$.
    Indeed, consider a clause $C$. We know that there exists at least one $j \in [3]$ such that $a^c_j$ moves through a vertex $u$ belonging in $G^x$, for a variable $x$. There are two cases to be analysed.

    \textbf{Case 1 ($\boldsymbol{u = u^x_{i,2}}$)}: In this case, by construction, we know that $x$ appears positively in the clause $C$.
    From the previous observations, it follows that $s_1 (a^c_j) = u = u^x_{i,2}$ and thus
    $s_1 (b^x_{i,1}) = w^x_{i}$ (as it can not be $u^x_{i,2}$). Therefore, by Claim~\ref{cl:makespan-2-variable-agents}, $s_1 (a^x_{1,1}) = v^x_{1,2}$ and $x$ has been set to true.

    \textbf{Case 2 ($\boldsymbol{u = v^x_{i,2}}$)}: In this case, by construction, we know that $x$ appears negatively in the clause $C$.
    From the previous observations, it follows that $s_1 (a^c_j) = u = v^x_{i,2}$ and thus
    $s_1 (a^x_{i,1}) = w^x_{i}$. Therefore, by Claim~\ref{cl:makespan-2-variable-agents}, $s_1 (a^x_{1,1})= w^x_{1}$ and $x$ has been set to false.

    In both cases above, the clause $C$ is satisfied. In other words, the assignment we defined is indeed satisfying $\phi$.

    The planarity of the graph $G$ follows from exactly the same arguments as these in proof of Theorem~\ref{thm:MAPFisNPC:PlanarMakespanDelta}.
\end{proof}

\section{Conclusion}
In this paper we studied the parameterised complexity of the \MAPF problem. The main takeaway message is that the problem is rather intractable. Indeed, the problem remains hard even on trees. Hence, the treewidth of the input graph is highly unlikely to yield an efficient algorithm (under standard theoretical assumptions). This suggests that the current heuristic-oriented approach followed by the community is, in some sense, optimal. On the positive side, we showed that there are various combinations of parameters that lead to efficient algorithms, such as the number of agents plus the makespan. These positive results could potentially lead to improvements in practice, which should be the subject of a dedicated future study. The first step towards this direction is to check whether our algorithms can be utilised in tandem with some state-of-the-art heuristic algorithm, in order to obtain an improved result.

\bibliographystyle{abbrv}
\bibliography{bibliography}

\end{document}